\documentclass{article}
\usepackage[dvipsnames]{xcolor}

\usepackage{dashbox}

\usepackage[dvipsnames]{xcolor}

\makeatletter
\@ifpackageloaded{hyperref}{
}{
  \usepackage{hyperref}
}
\makeatother

\hypersetup{colorlinks, citecolor=OliveGreen}

\usepackage{local-macros}
\reporttrue%

\usepackage{booktabs}
\usepackage{subcaption}
\usepackage{ifxetex}
\usepackage{microtype}
\usepackage{titling}

\ifxetex%
  \usepackage[no-math]{fontspec}
  \usepackage{mathspec}
  \usepackage{xltxtra}
  \usepackage{xunicode}
  \usepackage{titling}
  \usepackage{titlesec}
  \usepackage{abstract}

  \IfFontExistsTF{Minion Pro}{
    \setmainfont{Minion Pro}
  }{
  }

  \IfFontExistsTF{Cronos Pro}{
    \setsansfont{Cronos Pro}
  }{
  }

  \IfFontExistsTF{PragmataPro}{
    \setmonofont[Scale=0.9]{PragmataPro}
    \setmathsfont(Greek,Latin,Digits)[Scale=0.9]{PragmataPro}
    \setmathtt[Scale=0.9]{PragmataPro}
    \setmathrm[Scale=0.9]{PragmataPro}
    \setmathcal[Scale=0.9, StylisticSet=07]{PragmataPro}
    \renewcommand\ISET{\mathbf{SET}}
    \renewcommand\IFIN{\mathbf{FIN}}
    \renewcommand\IVar[1]{\mathbf{Var}_{#1}}
    \renewcommand\IVal{\mathbf{Val}}
    \renewcommand\ITm[1]{\mathbf{Prog}_{#1}}
    \renewcommand\ITl[1]{\mathbf{Tl}_{#1}}
    \renewcommand\TYPE{\mathbf{Type}}
    \renewcommand\LOC{\mathbf{Loc}}
    \renewcommand\ILam[1]{\mathbf{\lambda}\Parens*{#1}}
    \renewcommand\IKLam[1]{\mathbf{\lambda}_{\CLK}\Parens*{#1}}
    \renewcommand\ClkTopos{\mathbf{S}_{\CLK}}
    \renewcommand\CloCase[2]{\mathfrak{#1}\Parens*{#2}}
    \renewcommand\CloKFun[1]{\mathfrak{Fun}_{\CLK}\Parens*{#1}}
    \renewcommand\Conns[1]{\mathfrak{Conn}\Parens*{#1}}
  }{
  }

  \IfFontExistsTF{UnifrakturMaguntia}{
    \setmathfrak[Scale=1.1]{UnifrakturMaguntia}
  }{
  }

  \titleformat*{\section}{\Large\bfseries\sffamily}
  \titleformat*{\subsection}{\large\bfseries\sffamily}
  \titleformat*{\subsubsection}{\itshape\bfseries\sffamily}
  \titleformat*{\paragraph}{\itshape\bfseries\sffamily}

\else
  \usepackage[tt=false]{libertine}
  \usepackage[libertine]{newtxmath}
\fi

\allowdisplaybreaks%

\usepackage[section]{placeins}
\makeatletter
\AtBeginDocument{%
  \expandafter\renewcommand\expandafter\subsection\expandafter{%
    \expandafter\@fb@secFB\subsection%
  }%
}
\makeatother

\theoremstyle{plain}
\newtheorem{theorem}{Theorem}
\newtheorem{lemma}[theorem]{Lemma}
\newtheorem{corollary}[theorem]{Corollary}
\newtheorem{remark}[theorem]{Remark}

\theoremstyle{definition}
\newtheorem{definition}[theorem]{Definition}
\newtheorem{notation}[theorem]{Notation}

\jmsdraftfalse%

\begin{document}

\title{Guarded Computational Type Theory}
\author{
  Jonathan Sterling\thanks{\texttt{jmsterli@cs.cmu.edu}}\\
  {\small Carnegie Mellon University}
  \and
  Robert Harper\thanks{\texttt{rwh@cs.cmu.edu}}\\
  {\small Carnegie Mellon University}
}

\maketitle

\begin{abstract}
  Nakano's \emph{later} modality can be used to specify and define
recursive functions which are causal or synchronous; in concert with
a notion of clock variable, it is possible to also capture the
broader class of productive (co)programs. Until now, it has been
difficult to combine these constructs with dependent types in a way
that preserves the operational meaning of type theory and admits a
hierarchy of universes $\CT{\TyUniv{i}}$.

We present an operational account of guarded dependent type theory
with clocks called \ClockCTT{}, featuring a novel clock intersection
connective $\CT{\ClkIsect{k}{A}}$ that enjoys the clock irrelevance
principle, as well as a predicative hierarchy of universes
$\CT{\TyUniv{i}}$ which does not require any indexing in clock
contexts.
\ClockCTT{} is simultaneously a programming language with a rich
specification logic, as well as a computational metalanguage that
can be used to develop semantics of other languages and logics.

\end{abstract}

\section{Introduction}

In a functional programming language, every definable function is
continuous in the following sense: each finite quantity of output is
induced by some finite quantity of input. To make this more precise,
if we consider the case of stream transformers
$\Of{F}{\mathbb{S}\to\mathbb{S}}$, we can see that finite prefixes of
the output depend only on finite prefixes of the input:
\begin{equation}\label{eq:stream-continuity}
  \forall\Of{\alpha}{\mathbb{S}}.\
  \forall\Of{i}{\Nat}.\
  \exists\Of{n}{\Nat}.\
  \forall\Of{\beta}{\mathbb{S}}.\
  \EqPrefix{n}{\alpha}{\beta}
  \Rightarrow
  \IsEq{{F(\alpha)}_i}{{F(\beta)}_i}
\end{equation}

From a programming perspective, this can be rephrased in terms of
\emph{reads} and \emph{writes}: for each write, the program is
permitted to perform a finite but unbounded number of
reads.

\paragraph{Causality} Another possible class of functionals are the
ones that can be implemented by a program which performs at most one
read for every write. These are called the \emph{causal} functionals,
and in the case of stream transformers, they are characterized by the
following causality principle:
\begin{equation}\label{eq:stream-causality}
  \forall\Of{\alpha}{\mathbb{S}}.\
  \forall\Of{i}{\Nat}.\
  \forall\Of{\beta}{\mathbb{S}}.\
  \EqPrefix{i}{\alpha}{\beta}
  \Rightarrow
  \IsEq{{F(\alpha)}_i}{{F(\beta)}_i}
\end{equation}

In other words, causal programs are the ones whose reads and writes
proceed in lock-step. While we can surely carve out this class of
functionals using predicates like~(\ref{eq:stream-causality}) above,
it is actually possible to define a new notion of stream
$\mathbb{S}_{\blacktriangleright}$ such that all functionals
$\Of{F}{\mathbb{S}_{\blacktriangleright}\to\mathbb{S}_{\blacktriangleright}}$
are \emph{automatically} causal in the sense
of~(\ref{eq:stream-causality}). This kind of stream is called a
``guarded stream'', and we will use the term ``sequence'' to refer to
ordinary streams.

Whereas ordinary streams or sequences are usually formed as the
greatest solution to the isomorphism
$\IsIso{\mathbb{S}}{\Nat\times\mathbb{S}}$, the guarded streams are
formed using a special ``later modality'' $\blacktriangleright$ due to
\citeauthor{nakano:2000},\footnote{The notation $\bullet$ was
  originally used in~\citet{nakano:2000}.} solving the isomorphism
$\IsIso{\mathbb{S}_{\blacktriangleright}}{\Nat\times\ObjLater{}{\mathbb{S}_{\blacktriangleright}}}$. Modalities
of this kind usually enjoy at least the following principles:
\begin{mathpar}
  A\to\ObjLater{}{A}
  \and
  \IsIso{\ObjLater{}{\Parens{A\times{}B}}}{\Parens{\ObjLater{}{A}\times\ObjLater{}{B}}}
  \and
  \ObjLater{}{\Parens{A\to{}B}}\to\Parens{\ObjLater{}{A}\to\ObjLater{}{B}}
  \and
  \Parens{\ObjLater{}{A}\to{}A}\to A
\end{mathpar}

The ratio of reads and writes specified in the type of a stream
transformer can be modulated by adjusting the number of later
modalities in the input and the output to the function.

\paragraph{Nakano's modality in semantics}
What is remarkable about \citeauthor{nakano:2000}'s later modality is
that fixed points for functions $\Of{F}{\ObjLater{}{A}\to{}A}$ always
exist, without placing any restriction on $F$ (such as monotonicity or
positivity). Applied within a type-theoretic metalanguage, then, the
later modality induces solutions to recursive domain equations which
are not set-theoretically interpretable, such as the following classic
definition of semantic types for a programming language with mutable
store~\citep{appel-mellies-richards-vouillon:2007,birkedal-mogelberg-schwinghammer-stovring:2011}:
\[
  \IsIso{
    \TYPE
  }{
    \Parens*{
      \LOC
      \xrightarrow{\mathit{fin}}
      \ObjLater{}{\TYPE}
    }
    \to
    \Pow{
      \IVal
    }
  }
\]

The later modality captures and internalizes the basic features of
less abstract techniques like step-indexing, enabling more streamlined
definitions and proofs that elide the bureaucratic performance of
explicit indexing and monotonicity obligations. Today, modalities of
this kind are of the essence for modern program logics like
\textbf{\textsf{Iris}}~\citep{jung-et-al:2015}.

\paragraph{Programming applications}

The fact that functions $\Of{F}{\ObjLater{}{A}\to{}A}$ always have
fixed points has beneficial consequences for the practice of (total)
functional programming on infinite data. In particular, clumsy
syntactic guardedness conditions which ensure productivity (such as
those used in Coq~\citep{coq:reference-manual},
Agda~\citep{norell:2009} and Idris~\citep{brady:2013}) can be replaced
with type structure, enabling more compositional styles of
programming.\footnote{A very closely related idea, sized types, has
  been deployed in the Agda proof assistant for exactly this
  purpose~\citep{vezzosi:2015}.}

However, the later modality is too restrictive to be used on its own,
because it rules out the functions which are not causal; but
acausal functions on infinite data are perfectly sensible, and are
very common in the real world. Consider, for instance, the function
which drops every second element from a stream! To define this function, one
would need a way to delete the modality; but without suitable
restrictions, such an elimination principle would trivialize the
modality and render it useless.

To resolve this problem, \citeauthor{atkey-mcbride:2013} have
introduced a notion of abstract clock $\kappa$ to represent ``time
streams'' together with universal quantification $\forall\kappa$ over
clocks, replacing Nakano's modality with a clock-indexed family of
modalities $\ObjLater{\kappa}$~\citep{atkey-mcbride:2013}.

Defining the type of $\kappa$-guarded streams as the solution to the
equation
$\IsEq{\mathbb{S}_\kappa}{\Nat\times\ObjLater{\kappa}{\mathbb{S}_\kappa}}$,
it is possible to define the acausal function that drops every
other element of a stream, with type
$\Parens{\forall\kappa.\ \mathbb{S}_\kappa}\to\Parens{\forall\kappa.\
  \mathbb{S}_\kappa}$. The reason that this is possible is that their
calculus exhibits the isomorphism
$\IsIso{\Parens{\forall\kappa.\
    \ObjLater{\kappa}{A}}}{\Parens{\forall\kappa.\ A}}$, as well as a
\emph{clock irrelevance} principle:
$\IsEq{\Parens{\forall\kappa.\ A}}{A}$ assuming that $\kappa$ is not
free in $A$; we summarize the constructs of this calculus in
Figure~\ref{fig:atkey-mcbride-rules}.

\newcommand\AtMcOf[4]{{#1};{#2}\vdash{#3}:{#4}}

\begin{figure}
  \begin{mathpar}
    \Infer{
      \AtMcOf{\Delta}{\Gamma}{\Lambda\kappa.\ e}{\forall\kappa.\ A}
    }{
      \AtMcOf{\Delta,\kappa}{\Gamma}{e}{A}
    }
    \and
    \Infer{
      \AtMcOf{\Delta}{\Gamma}{e[\kappa']}{A[\kappa\hookrightarrow\kappa']}
    }{
      \AtMcOf{\Delta}{\Gamma}{e}{\forall\kappa.\ A}
      \\\\
      \Member{\kappa'}{\Delta}
      \\
      \NotIn{\kappa'}{\FreeClocks{A}}
    }
    \and
    \Infer{
      \AtMcOf{\Delta}{\Gamma}{\mathtt{pure} (e)}{\ObjLater{\kappa}{A}}
    }{
      \AtMcOf{\Delta}{\Gamma}{e}{A}
    }
    \and
    \Infer{
      \AtMcOf{\Delta}{\Gamma}{f\circledast{}e}{\ObjLater{\kappa}{B}}
    }{
      \AtMcOf{\Delta}{\Gamma}{f}{\ObjLater{\kappa}{\Parens{A\to{}B}}}
      \\\\
      \AtMcOf{\Delta}{\Gamma}{e}{\ObjLater{\kappa}{A}}
    }
    \and
    \Infer{
      \AtMcOf{\Delta}{\Gamma}{\mathtt{force} (e)}{\forall\kappa.\ A}
    }{
      \AtMcOf{\Delta}{\Gamma}{e}{\forall\kappa.\ \ObjLater{\kappa}{A}}
    }
    \and
    \Infer{
      \AtMcOf{\Delta}{\Gamma}{\mathtt{fix} (f)}{A}
    }{
      \AtMcOf{\Delta}{\Gamma}{f}{\ObjLater{\kappa}{A}\to{}A}
    }
    \\
    \IsEq{\Parens{\forall\kappa.\ A}}{A}\quad(\NotIn{\kappa}{\FreeClocks{A}})
    \and
    \IsEq{\forall\kappa.\ A\times{}B}{(\forall\kappa.\ A)\times(\forall\kappa.\ B)}
  \end{mathpar}
  \caption{Selection of rules from~\citet{atkey-mcbride:2013}.}\label{fig:atkey-mcbride-rules}
\end{figure}

\subsection{Dependent type theory and guarded recursion}

It has been surprisingly difficult to cleanly extend the account of
guarded recursion with clocks to a full-spectrum dependently typed
programming language which enjoys any combination of the following
properties:

\begin{enumerate}
\item \emph{Computational canonicity:} any closed element of type
  $\TyBool$ computes to either $\Tt$ or $\Ff$.
\item \emph{Simple universes:} a single predicative and cumulative
  hierarchy of universes $\TyUniv{i}$ closed under base types,
  dependent function types, dependent pair types, lower universes,
  \textbf{later modalities} and \textbf{clock quantifiers}.
\item \emph{Clock irrelevance:} if $k$ is not mentioned in $A$
  and $A$ is a type, then $\forall k.\ A$ is equal to
  $A$.\footnote{Depending on the specific type theory, it may be desirable to
  realize this principle either as an isomorphism or as a definitional
  equality.}
\end{enumerate}

However, a dependent type theory with support for guarded recursion
and clocks is desirable for multiple reasons; here, we have focused on
causality as a useful construct for developing types qua behavioral
specifications on program behavior, but there is also the potential to
use such a dependent type theory as a computational metalanguage for
developing and proving the semantics of other languages and logics,
vaporizing the highly-bureaucratic step-indexed Kripke Logical
Relations which usually must be employed.

The latter perspective is elaborated in the context of guarded
dependent type theory without clocks in~\citet{paviotti:2015} as well
as~\citet{bizjak-birkedal-miculan:2014}, and we anticipate that the
addition of clocks will enable further developments along these lines.

\subsection{Guarded Computational Type Theory}

We contribute a new extensional and behavioral dependent type
theory \ClockCTT{} (pronounced ``Guarded Computational Type Theory'')
for guarded recursion and clocks in the Nuprl
tradition~\citep{allen-et-al:2006}, enjoying the following
characteristics:
\begin{enumerate}
\item Operational semantics and an immediate canonicity result at base
  types.

\item A clock-indexed later modality $\CT{\CttLater{k}{A}}$ which
  requires no special syntax for introduction or destruction.

\item A decomposition of the clock quantifier
  from~\citet{bizjak-mogelberg:2017} into
    a parametric part $\CT{\ClkIsect{k}{A}}$ and a non-parametric part
    $\CT{\ClkProd{k}{A}}$. The former is an intersection, and enjoys the
    crucial clock irrelevance principle; the latter is the cartesian product of
    a clock-indexed family of sets (right adjoint to weakening).

\item A guarded fixed point combinator which can be assigned the type
  $\CT{(\CttLater{k}{A}\to A)\to A}$.

\item A predicative hierarchy of universes $\CT{\TyUniv{i}}$ closed
  under all the connectives, free of indexing by clock contexts.
\end{enumerate}

Our operational account and canonicity result
(Theorem~\ref{thm:canonicity}) means that \ClockCTT{} can be regarded
simultaneously as a programming language with a rich specification
logic, \emph{and} as a computational metalanguage for developing
operational and denotational semantics of other languages and logics.

\subsubsection*{Coq formalization and synthetic approach}
Using the Coq proof assistant, we have formalized the fragment of our
type theory that contains universes, dependent function and pair
types, booleans, the later modality, and the two clock quantifiers (intersection and product);
the full Coq development is available
in~\citet{sterling-harper:2018:coq}. Throughout this paper, theorems
and rules will be related to their Coq analogues using a reference
like \CoqRef{Module.theorem\_name}.

\ifreport%
The principal difference between our informal presentation and the Coq
formalization is that in the formalization of the formal term
language, we use De Bruijn indices for both variables and clock names,
whereas here we use concrete names for readability. This simplified
the lemmas that we needed to prove about syntax, and about the
elaboration of formal terms into programs.
\fi%

We have used Coq's type theory as a proxy for the internal language of
the presheaf topos that we develop herein, axiomatizing in Coq
whatever objects and principles come not from the standard type
theoretic constructions, but are instead imported into the system via
forcing. The entire construction of \ClockCTT{}, then, is carried out
within the internal language of the topos, an anti-bureaucratic
measure which has made an otherwise daunting formalization effort
feasible.

The idea of developing operational models of programming languages
within the internal language of a topos is not new; see for
instance~\citet{staton:2007}, \citet{bizjak-birkedal-miculan:2014} and
\citet{paviotti:2015}. However, we believe that ours is the first
instance of this technique being applied toward the development of
semantics for a full-spectrum dependent type theory.


\subsection*{Acknowledgments}
  We are thankful to Carlo Angiuli, Lars Birkedal, Ale\v{s} Bizjak, Jonas Frey,
  Daniel Gratzer, Adrien Guatto, Pieter Hofstra, Bas Spitters, Sam Staton, and
  Joseph Tassarotti for helpful discussions on the semantics of guarded
  recursion, clock names and universe hierarchies. Thanks to David Christiansen
  for his comments on a draft of this paper.

  The authors gratefully acknowledge the support of the Air Force
  Office of Scientific Research through MURI grant
  FA9550-15-1-0053. Any opinions, findings and conclusions
  or recommendations expressed in this material are those of the
  authors and do not necessarily reflect the views of the AFOSR.%

\section{Programming in \ClockCTT{}}\label{sec:programming}
Following the \emph{computational meaning-theoretic} tradition
initiated by \citet{martin-lof:1979}, and developed further in the
Nuprl project~\citep{allen-et-al:2006}, we build Guarded Computational
Type theory (\ClockCTT{}) on the basis of an untyped programming
language, whose syntax is summarized in Figure~\ref{fig:syntax}.

In this paper, we distinguish between the syntax of \CT{\texttt{formal
    terms}} and the language of \AT{\textbf{\textsf{programs}}}; formal terms
are used by clients of a formalism for type theory, whereas programs
are the things which are actually endowed with operational
meaning. For many languages, the difference between formal terms and
programs is not so great, but for us the difference is essential; to
avoid confusion, we distinguish between these levels using colors.

\begin{figure*}[t]
  \begin{minipage}[c]{1.0\linewidth}
    \[
      \begin{array}{rrlr}
        \GrmLine{clocks}{k}{k}
        \\
        \GrmLine{terms}{M,A}{
           x,
           \Lam{x}{M},
           \Lam{k}{M},
           M\ N,
           M\ k,
           \Pair{M}{N},
           \Fst{M},
           \Snd{M}
        }
        \\
        &&
        \FmtProds{
           \Fix{x}{M},
           \bigstar,\Tt,\Ff,\If{M}{N}{O}
        }
        \\
        &&
        \FmtProds{
           \Ze,\Su{M},\IfZe{M}{N}{x}{O},
           \Sup{M}{x}{N},
           \WRec{M}{x}{y}{z}{N}
        }
        \\
        &&
        \FmtProds{
           \DProd{x}{A}{B},
           \DSum{x}{A}{B},
           \WTy{x}{A}{B},
           \TyEqu{A}{M}{N}
        }
        \\
        &&
        \FmtProds{
           \ClkProd{k}{A},
           \ClkIsect{k}{A},
           \CttLater{k}{A}
        }
        \\
        &&
        \FmtProds{
           \TyVoid,\TyUnit,
           \TyBool,\TyNat,\TyUniv{i}
        }
        \\
        \GrmLine{clock contexts}{\Delta}{\cdot, {\Delta,k}}
        \\
        \GrmLine{variable contexts}{\Psi}{\cdot, {\Psi,x}}
        \\
        \GrmLine{typing contexts}{\Gamma}{\cdot, {\Gamma,x:A}}
      \end{array}
    \]
  \end{minipage}
  \caption{The syntax of formal terms in Guarded Computational Type
    Theory (\ClockCTT). Formal terms $\CT{M}$ are identified up to
    renamings of their bound variables; by convention, bound variables
    are always assumed fresh.}\label{fig:syntax}
\end{figure*}

\paragraph{Formal Terms} The grammar includes operators for both terms and
types, which are not distinguished syntactically in any way. Typehood, equality
and type membership are \emph{semantic} properties which will be imposed after
we propound the meaning explanation in Section~\ref{sec:meaning-explanation}.
We include syntax for dependent function types $\CT{\DProd{x}{A}{B}}$,
dependent pair types $\CT{\DSum{x}{A}{B}}$, wellordering types
$\CT{\WTy{x}{A}{B}}$, extensional equality types $\CT{\TyEqu{A}{M}{N}}$,
clock-indexed later modalities $\CT{\CttLater{k}{A}}$, clock product types
$\CT{\ClkProd{k}{A}}$, clock intersection types $\CT{\ClkIsect{k}{A}}$,
booleans, natural numbers, and a countable hierarchy of type universes
$\CT{\TyUniv{i}}$. We define the following derived forms for non-dependent
function and pair types:
\begin{mathpar}
  \Define{
    \CT{A\to{} B}
  }{
    \CT{\DProd{x}{A}{B}}
  }
  \and
  \Define{
    \CT{A\times{} B}
  }{
    \CT{\DSum{x}{A}{B}}
  }
\end{mathpar}

\paragraph{Forming fixed points and primitive recursors}

General fixed points can be programmed in \ClockCTT{} exactly as in
the untyped $\lambda$-calculus, but in order to simplify our
metatheorems we have provided a primitive fixed point operator
$\CT{\Fix{x}{M}}$. This can, for instance, be used to realize the induction
principle for the natural numbers.

When a function has type $\CT{\CttLater{\kappa}{A}\to{A}}$, its
\emph{guarded} fixed point always exists and has type
$\CT{A}$. Because \ClockCTT{} is dependently typed, it is very easy
for us to write a program that computes the type of guarded streams of
bits relative to a clock $\CT{k}$, using the fixed point operator in
concert with the later modality; and using the clock intersection
type, we can transform this into the type of infinite sequences of
bits:
\begin{align*}
  \CT{\Stream}
  &\in
  \CT{
    \ClkProd{k}{\TyUniv{i}}
  }
  \\
  \ADefine{
    \CT{\Stream}
  }{
    \CT{\Lam{k}{\Fix{A}{\TyBool\times\CttLater{k}{A}}}}
  }
  \\[6pt]
  \CT{\Sequence}
  &\in
  \CT{\TyUniv{i}}
  \\
  \ADefine{
    \CT{\Sequence{}}
  }{
    \CT{\ClkIsect{k}{\Stream\ k}}
  }
\end{align*}

We will see in Section~\ref{sec:revisiting-streams} that these
expressions are indeed types in \ClockCTT{}.


\section{Mathematical Meaning Explanation}\label{sec:construction}

In the type-theoretic tradition of Martin-L\"of, formal language is
endowed with computational meaning through what is called a ``meaning
explanation''; this style of definition, which was first deployed by
Martin-L\"of in his seminal paper \emph{Constructive Mathematics and
  Computer Programming}~\citep{martin-lof:1979}, is closely related to
PER semantics and the method of computability. This computational
perspective was developed to its fullest extent in Nuprl's
\CTT{}~\citep{constable:1986,allen-et-al:2006},
which adds a theory of computational congruence to the picture,
together with many new connectives including intersections, unions,
subset comprehensions, quotients and image types.

A meaning explanation provides a semantics for types as specifications
of the execution behavior of untyped programs. As such, the judgments
of type theory express the compliance of a program with a
specification, which can be of arbitrary quantifier complexity, and
will not generally be decidable.  Any implementation of type theory
involves, in one form or another, a formal system for deriving correct
judgments that is, by definition, recursively enumerable and often
decidable.

To achieve various properties that are desirable of a formal system
(sometimes including decidability), programs are often decorated with
type information that is not needed during execution. The meaning
explanation is, then, lifted to the formalism along an erasure map
$\Sem{-}$ that removes these decorations.

A similar, but more elaborate transformation of syntax (from \CT{formal
  terms} to \AT{programs}) is used here to facilitate the meaning
explanation for guarded type theory in terms of the settings of a
collection of clocks. During the verification of a program
specification, the value of a clock may change (for instance,
underneath the later modality); the most direct way to express this is
to explicitly formulate the meaning explanation using a Kripke or
presheaf-style semantics: a ``possible world'' consists of a
collection of clocks and their settings, and we require specifications
to account for the expansion of the world with new clocks and the
alterations of their settings.

Doing so tends to clutter the meaning explanation by distributing the
conditioning on clocks throughout the semantics, and disrupts a basic
principle of type theory in the Martin-L\"of tradition, which is that
types should do little more than internalize the structures which are
already present in the judgmental base.

An alternative, which we adopt here, is to formulate the semantics in
a presheaf topos $\ClkTopos$ which accounts all at once for clocks and
the passage of time, so that the specifications given by types are
implicitly conditioned on them. This conditioning, which is implicit
when viewed from inside the topos, can be externalized and made
explicit using the Kripke-Joyal forcing semantics of
$\ClkTopos$~\citep{maclane-moerdijk:1992}.

To ensure that programs evolve appropriately along the
transitions between clock worlds simultaneously with their
specifications, we introduce a kind of ``higher-order abstract
syntax'' which links clocks in programs directly to their meaning in
the presheaf topos, as elements of the presheaf of clocks
$\Of{\ClkObj}{\ClkTopos}$. The passage to this new kind of syntax at
the interface between the formalism and the semantics is managed by an
elaboration function $\Sem{-}$.

\subsection{The semantic universe $\ClkTopos$}\label{sec:semantic-universe}
We will develop our semantic universe as a presheaf topos called
$\ClkTopos$ over a category of clock contexts and clock context
morphisms. We will require the following things to exist in
$\ClkTopos$:
\begin{enumerate}
\item An object $\Of{\ClkObj}{\ClkTopos}$ of \emph{clock names}.
\item A family of logical modalities
  $\Later{{\kappa}}{{\phi}}$ for clock names
  $\Of{\kappa}{\ClkObj}$ and predicates ${\phi}$ in $\ClkTopos$.
\end{enumerate}

When we define $\ClkTopos$, we will arrange for the following
principles to hold in its internal logic:
\begin{align*}
  &\exists\Of{\kappa}{\ClkObj}.\ \top
  \IfReport{\tag{Theorem~\ref{thm:local-clock}}}{}
  \\
  &\forall\Of{\phi}{\Omega^\ClkObj}.\
  \Parens*{\ClkForall{\kappa}{\Later{\kappa}{\phi(\kappa)}}}
  \Rightarrow
  \ClkForall{\kappa}{\phi(\kappa)}
  \IfReport{\tag{Theorem~\ref{thm:forcing}}}{}
  \\
  &
  \ClkForall{\kappa}{
    \forall\Of{\phi}{\Omega}.\
    \phi\Rightarrow\Later{\kappa}{\phi}
  }
  \IfReport{\tag{Theorem~\ref{thm:later-unit}}}{}
  \\
  &\ClkForall{\kappa}{
    \forall\Of{\phi,\psi}{\Omega}.\
    \IsEq{
      \Later{\kappa}{
        \Parens{\phi\land\psi}
      }
    }{
      \Parens*{
        \Later{\kappa}{\phi}
        \land
        \Later{\kappa}{\psi}
      }
    }
  }
  \IfReport{\tag{Theorem~\ref{thm:later-cartesian}}}{}
  \\
  &\ClkForall{\kappa}{
    \forall\Of{\phi,\psi}{\Omega}.\
    \IsEq{
      \Later{\kappa}{
        \Parens{\phi\Rightarrow\psi}
      }
    }{
      \Parens*{
        \Later{\kappa}{\phi}
        \Rightarrow
        \Later{\kappa}{\psi}
      }
    }
  }
  \IfReport{\tag{Theorem~\ref{thm:later-commutes-with-implication}}}{}
  \\
  &\ClkForall{\kappa}{
    \forall\Of{\phi}{\Omega}.\
    \Parens*{\Later{\kappa}{\phi}\Rightarrow\phi}
    \Rightarrow
    \phi
  }
  \IfReport{\tag{Theorem~\ref{thm:loeb-induction}}}{}
\end{align*}

\ifreport%
We require one additional axiom to hold for any object
$\Of{Y}{\ClkTopos}$ which is \emph{total} and inhabited in a sense
that we will define
(Definitions~\ref{def:totality},~\ref{def:inhabitedness}), analogous
to the notion
from~\citet{birkedal-mogelberg-schwinghammer-stovring:2011}:
\else
We require one additional axiom to hold for any object
$\Of{Y}{\ClkTopos}$ which is \emph{total} and inhabited in a sense
analogous
to the notion
from~\citet{birkedal-mogelberg-schwinghammer-stovring:2011}:
\fi
\[
  \forall\Of{\kappa}{\ClkObj}.\
  \forall\Of{\phi}{\Omega^Y}.\
  \Later{\kappa}{
    \Parens*{
      \exists\Of{y}{Y}.\
      \phi(y)
    }
  }
  \Rightarrow
  \exists\Of{y}{Y}.\
  \Later{\kappa}{
    \phi(y)
  }
  \IfReport{\tag{Theorem~\ref{thm:total-yank-existential}}}{}
\]

To construct $\ClkTopos$ as a topos of presheaves, first define
$\Of{\FINPlus}{\CAT}$ as the free category with strictly associative
binary products generated by a single object; explicitly, objects of
$\FINPlus$ are $\IsEq{U}{\bullet^n}$ for $n>0$. A map
$\Of{f}{\bullet^n\to\bullet^m}$ is a vector of projections, but can
dually be regarded as a function between finite sets
$\Nat_{<m}\to \Nat_{<n}$.

Observe that the opposite category ${\OpCat{\FINPlus}}$ is a skeleton
of the category of non-empty finite sets and all functions between
them.
$\FINPlus$ is also a full subcategory of $\Of{\FIN}{\CAT}$, the free
strict cartesian category generated by a single object (whose opposite
is likewise a skeleton of the category of finite sets and all maps
between them).

\ifreport%
\begin{remark}\label{rem:sheaves}
  The category of presheaves ${\Psh{\FINPlus}}$ is equivalent to
  the sheaf subcategory of ${\Psh{\FIN}}$ under the coverage
  generated by singleton families of
  epimorphisms~\citep{staton:2007}. This sheaf subcategory is
  completely analogous to the Schanuel topos (i.e.\ the category of
  nominal sets), except that names are subject to
  identification/contraction. When names are used to represent clocks,
  this phenomenon has been referred to as
  ``synchronization'' by~\citet{bizjak-mogelberg:2015}.
\end{remark}
\fi%

Define the presheaf of clock names
$\Of{\mathcal{N}}{\Psh{\FINPlus}}$ as the representable functor
$\Yoneda\Parens*{\bullet^1}$. Next, define a functor
$\Of{\CLK[-]}{\FINPlus\to\POS}$ (with $\POS$ the category of partially
ordered sets) which will interpret assignments of \emph{times} to
clock names:
\begin{align*}
  \AOf{\CLK[-]}{\FINPlus\to\POS}
  \\
  \ADefine{
    \CLK\Squares{\Of{U}{\FINPlus}}
  }{
    \omega^{\mathcal{N}(U)}
  }
  \\
  \ADefine{
    \CLK
    \Squares{\Of{f}{V\to U}}
    \Parens{\Of{\partial_V}{\omega^{\mathcal{N}(V)}}}
  }{
    \Parens{\Of{\kappa}{\mathcal{N}(U)}}
    \mapsto
    \partial_V\Parens*{f^*\kappa}
  }
\end{align*}
Thinking of elements of $\FINPlus$ as signifying finite and non-empty
cardinalities of clock names, the action of $\CLK[-]$ on objects takes
such a cardinality $\Of{U}{\FINPlus}$ to the $U$-fold product of the
poset $\omega$, ordered pointwise: in other words, it assigns the
amount of ``time left'' to each clock.

Finally, using the covariant Grothendieck
  construction~\citep{crole:1993} we can build the total category
$\Define{\HIBox{\Of{\CLK}{\CAT}}}{\int^{\FINPlus}\CLK[-]}$ in the
following way. Objects are pairs
$(\Of{U}{\FINPlus},\Of{\partial_U}{\CLK[U]})$, i.e.\ collections of
clock names together with an assignment; morphisms
$\Of{f}{(V,\partial_V)\to(U,\partial_U)}$ are
${\FINPlus}$-morphisms $\Of{f}{V\to U}$ such that
$\IsLEQ{\CLK[f](\partial_V)}{\partial_U}$ in ${\CLK[U]}$.
At this time it will be helpful to impose some notation: we will write
$\Of{\ell}{\CLK\to\FINPlus}$ for the induced projection functor, and
we will use boldface letters ${\mathbf{U}},{\mathbf{V}}$
to range over objects
$\HIBox{\Of{{(U,\partial_U)},{(V,\partial_V)}}{\CLK}}$.

\paragraph{The semantic universe $\ClkTopos$}
Finally, we define our semantic universe as the presheaf topos
$\Define{\ClkTopos}{\Psh{\CLK}}$. This ``topos of clocks'' defined
above inherits a rich internal logic which corresponds to a
combination of cartesian/structural nominal logic\footnote{That is,
  the logic of \emph{nominal substitution
    sets}~\citep{staton:2007,gabbay-hofmann:2008}.} and guarded
recursion.

The topos $\ClkTopos$ is related to the models considered
by~\citet{bizjak-mogelberg:2015}, except that rather than constructing
a family of presheaf toposes fibered over clock contexts, we combine
clock contexts with time assignments into a single base category, and
take the topos of presheaves over that; our topos is nearly identical
to the presheaf category considered independently
in~\citet{bizjak-mogelberg:2017}.

One minor difference between our model and those of
\citeauthor{bizjak-mogelberg:2017} is that in order to close the
internal logic of $\ClkTopos$ under the clock irrelevance axiom
described above, we decided to rule out empty clock contexts; this
condition is equivalent to taking a sheaf subtopos of the presheaves
over \emph{all} clock contexts.

\paragraph{The object of clock names}

We need to exhibit an object in the presheaf topos ${\ClkTopos}$ whose
elements are the ``available'' clock \emph{names} (without regard to
their time assignments). First observe that the representable object
$\mathcal{N}$ plays exactly this role in the category
$\Psh{\FINPlus}$: at clock context $\bullet^n$ it consists in the set
of morphisms $\bullet^n\to\bullet^1$, which has cardinality
$n$. However, this object resides in the wrong topos, since we need to
define an object $\Of{\ClkObj}{\ClkTopos}$.
To achieve this, we use the reindexing functor
$\Of{\ell^*}{\Psh{\FINPlus}\to\ClkTopos}$ induced by precomposing the
projection $\Of{\ell}{\ClkTopos\to\FINPlus}$, defining
$\Define{\ClkObj}{\ell^*\mathcal{N}}$.

\paragraph{Notations and morphisms}
We write $\mathbf{U}[\kappa\mapsto n]$ to mean
$(U, \partial_U[\kappa\mapsto n])$, where
$\partial_U[\kappa\mapsto n]$ means the adjustment to $\partial_U$
which replaces $\partial_U(\kappa)$ with $n$. Finally, for the map
that increments the time assigned to a clock, we write
$\Of{[\kappa\pluseq1]}{\mathbf{U}\to\mathbf{U}[\kappa\mapsto\partial_U(\kappa)+1]}$.

\paragraph{Defining the ${\Later{\kappa}{}}$ modalities}

We define the ${\Later{\kappa}{}}$ modalities by their forcing clause
in the Kripke-Joyal semantics of $\ClkTopos$:\footnote{Usually the
  forcing clauses should be taken as theorems rather than as
  definitions. However, in a Grothendieck topos, it is possible to
  define a subobject by its forcing clause: the result is well-defined
  when the definition is monotone (and also local, in the case of
  sheaf toposes).}
\[
  \DefineJdg{
    \Forces{\mathbf{U}}{\Later{\kappa}{\phi(\alpha)}}
  }{
    \begin{cases}
      \top &{\normalcolor\textbf{if}\ \ \IsEq{\partial_U(\kappa)}{0}}
      \\
      \Forces{
        \mathbf{U}[\kappa\mapsto n]
      }{
        \phi({[\kappa\pluseq1]}^* \alpha)
      } &{\normalcolor\textbf{if} \ \ \IsEq{\partial_U(\kappa)}{n+1}}
    \end{cases}
  }
\]

By a similar definition, it is possible to define an analogous
operator in the internal type theory of $\ClkTopos$, i.e.\ a fibered
endofunctor
$\Of{\ObjLater{}}{\Slice{\ClkTopos}{X\times\ClkObj}\to\Slice{\ClkTopos}{X\times\ClkObj}}$;
however, we have only needed the logical modality in our construction.

All the other forcing clauses are completely standard; for a reference
on Kripke-Joyal forcing, see~\citet{maclane-moerdijk:1992}.

\subsection{Programming language and operational semantics}

In Section~\ref{sec:programming} (Figure~\ref{fig:syntax}) we gave a
grammar for the \CT{formal terms} of \ClockCTT{}; however, in our
semantics, we employ a second notion of syntax which is constructed as
an inductive definition internal to $\ClkTopos$; this is the language
of \AT{programs}, and differs from the syntax of formal terms in two
respects:
\begin{enumerate}
\item Clocks in programs are imported directly from the metatheoretic
  object of clocks $\Of{\ClkObj}{\ClkTopos}$; so the family of
  operators $\AT{\CttLater{\kappa}{-}}$ is indexed in
  $\Of{\kappa}{\ClkObj}$ in exactly the same way that
  $\AT{\ITyUniv{i}}$ is indexed in $\Of{i}{\Nat}$.
\item The binding of clocks (such as in the clock intersection
  operator) is represented using the exponential
  $\Of{-^\ClkObj}{\ClkTopos\to\ClkTopos}$.\footnote{While this
    construction cannot be called ``ordinary syntax'', it is an
    inductive definition that can be built up explicitly using the
    fact that $\ClkTopos$ models indexed
    W-types~\citep{moerdijk:2000}.}
\end{enumerate}

\begin{remark}[Generalized Syntax]
  The idea of using the exponential of the metalanguage in the syntax
  of a programming language is not new.
  Infinitary notions of program syntax can be traced back as far as
  Brouwer's $\digamma$-inference in the justification of the Bar
  Principle~\citep{brouwer:1981}, and have more recently been
  developed in Nuprl semantics~\citep{rahli-bickford-constable:2017}, as well as in
  the context of higher-order focusing~\cite{zeilberger:2009}.
\end{remark}

We will define the inductive family $\ITm{n}$ of \emph{programs
  with $n$ free variables} in $\ClkTopos$ using an internal inductive
definition, summarized in Figure~\ref{fig:programs}.

\ifreport%
\begin{figure*}
  \input{figures/programs}
  \caption{The inductive definition of the programs with $n$ free
    variables $\Of{\ITm{n}}{\ClkTopos}$.}\label{fig:programs}
\end{figure*}
\begin{figure*}
  \input{figures/opsem}
  \caption{Structural operational semantics of closed \ClockCTT{} programs.}\label{fig:opsem}
\end{figure*}
\else%
\begin{figure*}
  \input{figures/programs-opsem-short}
  \caption{An illustrative fragment of the inductive definition of the
    programs with $n$ free variables $\Of{\ITm{n}}{\ClkTopos}$,
    and their operational semantics.}\label{fig:programs}
\end{figure*}
\fi%

\ifreport

  \paragraph{Substitution structure}
  Writing $\ISET$ to mean the
  internal category of small sets in $\ClkTopos$, observe that
  $\IVar{-}$ can be regarded as an internal functor from $\IFIN$ to
  $\ISET$, where $\IFIN$ is the internal category of finite cardinals
  and all functions between them. We can equip $\ITm{-}$ with the
  structure of a relative monad on
  $\Of{\IVar{-}}{\IFIN\to\ISET}$~\citep{altenkirch-chapman-uustalu:2010}.

  The unit of the relative monad is the injection of variables
  $\AT{\Var{(-)}}$; its Kleisli extension implements substitutions
  $\IsITm{\Subst{\gamma}{M}}{n}$ for $\IsITm{M}{m}$ and
  $\Of{\AT{\gamma}}{\ITm{n}^{\IVar{m}}}$. We omit the definition of the
  Kleisli extension because it is completely standard.

\else
  Our syntax forms a substitution algebra, and we write
  $\Of{\AT{\Var{(-)}}}{\IVar{n}\to\ITm{n}}$ for the injection of variables into
  terms; we write $\IsITm{\Subst{\gamma}{M}}{n}$ for the action of the
  substitution $\Of{\AT{\gamma}}{\ITm{n}^{\IVar{m}}}$ in $\IsITm{M}{m}$.
\fi

\paragraph{Internal operational semantics}
Programs are endowed with operational meaning through the definition of a transition system,
\ifreport%
summarized in Figure~\ref{fig:opsem}.
\else%
an illustrative fragment of which we present in Figure~\ref{fig:programs}.
\fi%
This defines predicates $\Of{\HIBox{\Val{-}}}{\Pow{\ITm{0}}}$ and
$\Of{\HIBox{\Step{-}{-}}}{\Pow{\ITm{0}\times\ITm{0}}}$ in
$\ClkTopos$. Write $\Of{\IVal}{\ClkTopos}$ for the subobject
$\MkSet{\IsITm{M}{0}\mid\Val{M}}$.

Write $\HIBox{\StepStar{-}{-}}$ for the reflexive-transitive closure
of $\HIBox{\Step{-}{-}}$.  We now define approximation and
computational equivalence judgments
$\Of{\HIBox{\ClosedApprox{-}{-}},\HIBox{\ClosedSq{-}{-}}}{\Pow{\ITm{0}\times\ITm{0}}}$
respectively for closed programs as follows:
\begin{align*}
  \ADefine{
    \HIBox{\ClosedApprox{M_0}{M_1}}
  }{
    \forall\Of{\AT{M_v}}{\IVal}.\
    \HIBox{\StepStar{M_0}{M_v}}
    \Rightarrow
    \HIBox{\StepStar{M_1}{M_v}}
  }
  \\
  \ADefine{
    \HIBox{\ClosedSq{M_0}{M_1}}
  }{
    \HIBox{\ClosedApprox{M_0}{M_1}}
    \land
    \HIBox{\ClosedApprox{M_1}{M_0}}
  }
\end{align*}

The latter is extended to a computational equivalence judgment for
open programs
$\Of{\HIBox{\OpenSq{n}{-}{-}}}{\Pow{\ITm{n}\times\ITm{n}}}$ by
quantifying over total substitutions.
\begin{align*}
  \ADefine{
    \HIBox{\OpenSq{n}{M_0}{M_1}}
  }{
    \forall\Of{\AT{\gamma}}{\ITm{0}^n}.\
    \HIBox{
      \ClosedSq{
        \Subst{\gamma}{M_0}
      }{
        \Subst{\gamma}{M_1}
      }
    }
  }
\end{align*}

It would be desirable to extend this relation to a theory of
computational congruence, as pioneered by~\citet{howe:1989}; however,
for our immediate purposes it has sufficed to require types only to
respect the approximation relation defined above.

\begin{definition}[Computational PERs]\label{def:cper}
  A partial equivalence relation is a binary relation which is both
  symmetric and transitive. Such a relation $\mathcal{R}$ on $\ITm{0}$
  is called \emph{computational} when it respects approximation in the
  following sense: if
  $\Member{\Parens*{\AT{M_0}, \AT{M_1}}}{\mathcal{R}}$ and
  $\ClosedApprox{M_0}{M_0'}$, then
  $\Member{\Parens*{\AT{M_0'}, \AT{M_1}}}{\mathcal{R}}$.
\end{definition}

\paragraph{Telescopes}

To capture the syntax of contexts and we define the inductive family
$\ITl{n}$ of \emph{telescopes of length $n$} as follows:
\begin{mathparpagebreakable}
  \Infer{
    \IsITl{\cdot}{0}
  }{
  }
  \and
  \Infer{
    \IsITl{\Gamma. A}{n+1}
  }{
    \IsITl{\Gamma}{n}
    \\
    \IsITm{A}{n}
  }
\end{mathparpagebreakable}

\paragraph{Elaborating terms}
We now sketch the elaboration of the \CT{program terms} of
Section~\ref{sec:programming} into \AT{programs}; approximately, a
term $\CT{M}$ with free formal clock variables $\CT{\Delta}$ and free
term variables $\CT{\Psi}$ will be elaborated to a morphism
$\Of{\SemTm{\Delta}{\Psi}{M}}{\ClkObj^{\Dom{\CT{\Delta}}}\to\ITm{\Dom{\CT{\Psi}}}}$.

\begin{notation}
  When $\CT{\Psi}$ is a list, we write $\Dom{\CT{\Psi}}$ for its
  length, and we write $\CT{\Psi}[\CT{x}]$ for the index
  $i<\Dom{\CT{\Psi}}$ of the element $\CT{x}$ in $\CT{\Psi}$,
  presupposing $\CT{\Psi}\ni\CT{x}$.
\end{notation}

\ifreport%
\begin{align*}
  \SemTm{\Delta}{\Psi}{x}\varrho
  &= \AT{\Var{\CT{\Psi}[\CT{x}]}}
  \\
  \SemTm{\Delta}{\Psi}{\Lam{x}{M}}\varrho
  &=
  \AT{
    \ILam{
      \SemTm{\varrho}{\Psi,x}{M}\varrho
    }
  }
  \\
  \SemTm{\Delta}{\Psi}{\Lam{k}{M}}\varrho
  &=
  \AT{
    \IKLam{
      \MetaLam{\kappa}{
        \SemTm{\Delta,k}{\Psi}{M}(\varrho,\kappa)
      }
    }
  }
  \\
  \SemTm{\Delta}{\Psi}{M_0\ M_1}\varrho
  &=
  \AT{
    \Parens*{
      \SemTm{\Delta}{\Psi}{M_0}\varrho
    }
    \Parens*{
      \SemTm{\Delta}{\Psi}{M_1}\varrho
    }
  }
  \\
  \SemTm{\Delta}{\Psi}{M\ k}\varrho
  &=
  \AT{
    \Parens{
      \SemTm{\Delta}{\Psi}{M}\varrho
    }%
    \Parens{
      \rho_{\Delta[\CT{k}]}
    }
  }
  \\
  \SemTm{\Delta}{\Psi}{\Pair{M_0}{M_1}}\varrho
  &=
  \AT{
    \IPair{
      \SemTm{\Delta}{\Psi}{M_0}\varrho
    }{
      \SemTm{\Delta}{\Psi}{M_1}\varrho
    }
  }
  \\
  \SemTm{\Delta}{\Psi}{\Fst{M}}\varrho
  &=
  \AT{
    \IFst{
      \Parens*{\SemTm{\Delta}{\Psi}{M}{\varrho}}
    }
  }
  \\
  \SemTm{\Delta}{\Psi}{\Snd{M}}\varrho
  &=
  \AT{
    \ISnd{
      \Parens*{\SemTm{\Delta}{\Psi}{M}{\varrho}}
    }
  }
  \\
  \SemTm{\Delta}{\Psi}{\Sup{M}{x}{N}}\varrho
  &=
  \AT{
    \ISup{
      \SemTm{\Delta}{\Psi}{M}\varrho
    }{
      \SemTm{\Delta}{\Psi,x}{N}\varrho
    }
  }
  \\
  \SemTm{\Delta}{\Psi}{\WRec{M}{x}{y}{z}{N}}\varrho
  &=
  \AT{
    \IWRec{
      \SemTm{\Delta}{\Psi}{M}\varrho
    }{
      \SemTm{\Delta}{\Psi,x,y,z}{N}\varrho
    }
  }
  \\
  \SemTm{\Delta}{\Psi}{\bigstar}\varrho
  &= \AT{\bigstar}
  \\
  \SemTm{\Delta}{\Psi}{\Tt}\varrho
  &= \AT{\ITt}
  \\
  \SemTm{\Delta}{\Psi}{\Ff}\varrho
  &= \AT{\IFf}
  \\
  \SemTm{\Delta}{\Psi}{\If{M_b}{M_t}{M_f}}\varrho
  &=
  \AT{
    \IIf{
      \SemTm{\Delta}{\Psi}{M_b}\varrho
    }{
      \SemTm{\Delta}{\Psi}{M_t}\varrho
    }{
      \SemTm{\Delta}{\Psi}{M_f}\varrho
    }
  }
  \\
  \SemTm{\Delta}{\Psi}{\Ze}\varrho
  &=\AT{\IZe}
  \\
  \SemTm{\Delta}{\Psi}{\Su{M}}\varrho
  &=
  \AT{\ISu{\SemTm{\Delta}{\Psi}{M}\varrho}}
  \\
  \SemTm{\Delta}{\Psi}{\IfZe{M_n}{M_z}{x}{M_s}}\varrho
  &=
  \AT{
    \IIfZe{
      \SemTm{\Delta}{\Psi}{M_n}\varrho
    }{
      \SemTm{\Delta}{\Psi}{M_z}\varrho
    }{
      \SemTm{\Delta}{\Psi,x}{M_s}\varrho
    }
  }
  \\
  \SemTm{\Delta}{\Psi}{\DProd{x}{A}{B}}\varrho
  &=
  \AT{
    \IPi{
      \SemTm{\Delta}{\Psi}{A}\varrho
    }{
      \SemTm{\Delta}{\Psi,x}{B}\varrho
    }
  }
  \\
  \SemTm{\Delta}{\Psi}{\DSum{x}{A}{B}}\varrho
  &=
  \AT{
    \ISg{
      \SemTm{\Delta}{\Psi}{A}\varrho
    }{
      \SemTm{\Delta}{\Psi,x}{B}\varrho
    }
  }
  \\
  \SemTm{\Delta}{\Psi}{\WTy{x}{A}{B}}\varrho
  &=
  \AT{
    \IWTy{
      \SemTm{\Delta}{\Psi}{A}\varrho
    }{
      \SemTm{\Delta}{\Psi,x}{B}\varrho
    }
  }
  \\
  \SemTm{\Delta}{\Psi}{\TyEqu{A}{M_0}{M_1}}\varrho
  &=
  \AT{
    \ITyEqu{
      \SemTm{\Delta}{\Psi}{A}\varrho
    }{
      \SemTm{\Delta}{\Psi}{M_0}\varrho
    }{
      \SemTm{\Delta}{\Psi}{M_1}\varrho
    }
  }
  \\
  \SemTm{\Delta}{\Psi}{\CttLater{k}{A}}\varrho
  &=
  \AT{
    \CttLater{
      \varrho_{\Delta[k]}
    }{
      \SemTm{\Delta}{\Psi}{A}\varrho
    }
  }
  \\
  \SemTm{\Delta}{\Psi}{\ClkProd{k}{A}}\varrho
  &=
  \AT{
    \IClkProd{
      \Parens*{
        \MetaLam{\kappa}{
          \SemTm{\Delta,k}{\Psi}{A}(\varrho,\kappa)
        }
      }
    }
  }
  \\
  \SemTm{\Delta}{\Psi}{\ClkIsect{k}{A}}\varrho
  &=
  \AT{
    \IClkIsect{
      \Parens*{
        \MetaLam{\kappa}{
          \SemTm{\Delta,k}{\Psi}{A}(\varrho,\kappa)
        }
      }
    }
  }
  \\
  \SemTm{\Delta}{\Psi}{\TyVoid}\varrho
  &= \AT{\ITyVoid}
  \\
  \SemTm{\Delta}{\Psi}{\TyUnit}\varrho
  &= \AT{\ITyUnit}
  \\
  \SemTm{\Delta}{\Psi}{\TyBool}\varrho
  &= \AT{\ITyBool}
  \\
  \SemTm{\Delta}{\Psi}{\TyNat}\varrho
  &= \AT{\ITyNat}
  \\
  \SemTm{\Delta}{\Psi}{\TyUniv{i}}\varrho
  &= \AT{\ITyUniv{i}}
\end{align*}

 \else%
We present here only a few of the most illustrative cases; the
remainder of the elaboration can be found in
technical report, and in our Coq
formalization~\citep{sterling-harper:2018:coq}.
\begin{align*}
  \SemTm{\Delta}{\Psi}{x}\varrho
  &= \AT{\Var{\CT{\Psi}[\CT{x}]}}
  \\
  \SemTm{\Delta}{\Psi}{\Lam{x}{M}}\varrho
  &=
  \AT{
    \ILam{
      \SemTm{\varrho}{\Psi,x}{M}\varrho
    }
  }
  \\
  \SemTm{\Delta}{\Psi}{\Lam{k}{M}}\varrho
  &=
  \AT{
    \IKLam{
      \MetaLam{\kappa}{
        \SemTm{\Delta,k}{\Psi}{M}(\varrho,\kappa)
      }
    }
  }
  \\
  \SemTm{\Delta}{\Psi}{M\ k}\varrho
  &=
  \AT{
    \Parens{
      \SemTm{\Delta}{\Psi}{M}\varrho
    }%
    \Parens{
      \rho_{\Delta[\CT{k}]}
    }
  }
  \\
  \SemTm{\Delta}{\Psi}{\CttLater{k}{A}}\varrho
  &=
  \AT{
    \CttLater{
      \varrho_{\Delta[\CT{k}]}
    }{
      \SemTm{\Delta}{\Psi}{A}\varrho
    }
  }
  \\
  \SemTm{\Delta}{\Psi}{\ClkProd{k}{A}}\varrho
  &=
  \AT{
    \IClkProd{
      \Parens*{
        \MetaLam{\kappa}{
          \SemTm{\Delta,k}{\Psi}{A}(\varrho,\kappa)
        }
      }
    }
  }
  \\
  \SemTm{\Delta}{\Psi}{\ClkIsect{k}{A}}\varrho
  &=
  \AT{
    \IClkIsect{
      \Parens*{
        \MetaLam{\kappa}{
          \SemTm{\Delta,k}{\Psi}{A}(\varrho,\kappa)
        }
      }
    }
  }
\end{align*}
\fi

\paragraph{Elaborating contexts}

Next, we elaborate contexts $\CT{\Gamma}$ with free formal clock variables
$\CT{\Delta}$ as morphisms
$\Of{\SemTl{\Delta}{\Gamma}}{\ClkObj^{\Dom{\CT{\Delta}}}\to\ITl{\Dom{\CT{\Gamma}}}}$,
writing $\pi(\CT{\Gamma})$ for the sequence $\CT{\vec{x_i}}$ when
$\IsEq{\CT{\Gamma}}{\overrightarrow{\CT{x_i:A_i}}}$.
\begin{align*}
  \SemTl{\Delta}{\cdot}\varrho
  &=
  \AT{\cdot}
  \\
  \SemTl{\Delta}{\Gamma,x:A}\varrho
  &=
  \AT{
    \Parens*{\SemTl{\Delta}{\Gamma}\varrho}.
    \Parens*{\SemTm{\Delta}{\pi(\Gamma)}{A}\varrho}
  }
\end{align*}

To save space, we may write $\Sem{M}$ or $\Sem{\Gamma}$ for the
elaboration of a term or a context respectively, when the
parameters are obvious.

\subsection{Full type system hierarchy}

At a high level, a \emph{type system} in the sense of
\citet{allen:1987:thesis} is an object which distinguishes some
programs as types, and specifies what programs will be the elements of
those types, and when they will be considered equal.
Writing $\Rel{X}$ for $\Pow{X\times{}X}$, we define a \emph{candidate
  type system} to be a relation
$\Of{\TS}{\Pow{\ITm{0}\times\Rel{\ITm{0}}}}$ in $\ClkTopos$. We will
write $\PTSObj$ for the collection of such candidate type systems,
i.e.\
$\Define{\HIBox{\Of{\PTSObj}{\ClkTopos}}}{\Pow{\ITm{0}\times\Rel{\ITm{0}}}}$.

Let us now define notation for some assertions about candidate type
systems $\Of{\TS}{\PTSObj}$:
\begin{align*}
  \ADefineJdg{
    \WithTS{\TS}{\SimType{A}{B}}
  }{
    \exists\Of{\mathcal{A}}{\Rel{\ITm{0}}}.\
    \Member{(\AT{A},\mathcal{A})}{\TS}
    \land
    \Member{(\AT{B},\mathcal{A})}{\TS}
  }
  \\
  \ADefineJdg{
    \WithTS{\TS}{\SimMem{M_0}{M_1}{A}}
  }{
    \exists\Of{\mathcal{A}}{\Rel{\ITm{0}}}.\
    \Member{(\AT{A},\mathcal{A})}{\TS}
    \land
    \Member{(\AT{M_1}, \AT{M_2})}{\mathcal{A}}
  }
\end{align*}

A candidate type system $\Of{\TS}{\PTSObj}$ can have the following
characteristics:
\begin{enumerate}
\item It is called \emph{extensional} if it is the graph of a partial
  function $\ITm{0}\rightharpoonup\Rel{\ITm{0}}$.
\item It is called \emph{computational PER-valued} if whenever
  $\Member{\Parens*{\AT{A},\mathcal{A}}}{\tau}$, the relation
  $\mathcal{A}$ is a computational PER (see
  Definition~\ref{def:cper}).
\item It is called \emph{type-computational} when, if
  $\Member{\Parens*{\AT{A},\mathcal{A}}}{\tau}$ and
  $\ClosedApprox{A}{A'}$, then also
  $\Member{\Parens*{\AT{A'},\mathcal{A}}}{\tau}$.
\end{enumerate}

Finally a candidate type system is called a \emph{type system} if it is
extensional, computational PER-valued, and type-computational. We
write $\Of{\TSObj}{\ClkTopos}$ for the collection of such type
systems.

\paragraph{Sequents and functionality}

Next, we briefly sketch the meaning of type functionality sequents
$\HypSimType{\Gamma}{A_0}{A_1}$ and functionality sequents
$\HypSimMem{\Gamma}{M_0}{M_1}{A}$ using a simple notion of
functionality derived from \citet{martin-lof:1979}, with respect to
any candidate type system $\Of{\TS}{\PTSObj}$.

When $\IsITl{\Gamma}{n}$ is a telescope and
$\Of{\AT{\gamma_0},\AT{\gamma_1}}{\ITm{0}^n}$ are sequences of programs, we define
similarity of instantiations
$\SimEnv{\gamma_0}{\gamma_1}{\Gamma}$ by
recursion on
$\Gamma$. $\SimEnv{\cdot}{\cdot}{\cdot}$ is
true, and
$\SimEnv{\gamma_0. M_0}{\gamma_1. M_1}{\Gamma. A}$
is true when both
$\SimEnv{\gamma_0}{\gamma_1}{\Gamma}$ and
$\SimMem{\Subst{\gamma_0}{M_0}}{\Subst{\gamma_1}{M_1}}{\Subst{\gamma_0}{A}}$
are true.

Open type similarity $\HypSimType{\Gamma}{A_0}{A_1}$ is true
when for all instantiations $\SimEnv{\gamma_0}{\gamma_1}{\Gamma}$, we
have
$\SimType{\Subst{\gamma_0}{A_0}}{\Subst{\gamma_1}{A_1}}$. Likewise,
open member smilarity $\HypSimMem{\Gamma}{M_0}{M_1}{A}$ is
true when for all such instantiations, we have
$\SimMem{\Subst{\gamma_0}{M_0}}{\Subst{\gamma_1}{M_1}}{\Subst{\gamma_0}{A}}$.

Finally, context validity $\IsCtx{\Gamma}$ is
given by recursion on $\Gamma$ using open type similarity in the
inductive case.

\subsection{Closure under type formers other than universes}

Next, we will show how to \emph{close} a candidate type system under
the type formers of \ClockCTT{}, namely booleans, natural numbers,
dependent functions types, dependent pair types, equality types, later
modalities, clock intersection types and universes.

The simplest way to carry out this construction, as pioneered
by~\citet{crary:1998} and formalized by~\citet{anand-rahli:2014}, is
to use an inductive definition of a closure operator
$\Of{\TSClose{-}}{\PTSObj\to\PTSObj}$ on candidate type
systems. However, this method does not immediately extend to the type
systems that we consider in this paper, because it is not clear how to
fit the clause for the \emph{later modality} into the usual schemata
for inductive definitions based on strictly positive signatures.

Therefore, as advocated by~\citet{allen:1987:thesis}, we will build up
our closure operator manually by taking the least fixed point of a
monotone operator on candidate type systems; this construction can be
carried out in any topos, because the Knaster-Tarski theorem
guarantees a least fixed point for any monotone operator on a complete
lattice~\citep{davey-priestly:1990}.

First, we define some notation for closing relations and type systems
under evaluation to canonical form:
\begin{align*}
  \ValClo{-}&:\Rel{\ITm{0}}\to\Rel{\ITm{0}}
  \\
  \ADefine{
    \ValClo{\mathcal{A}}
  }{
    \SetCompr{
      \Parens*{\AT{M_0}, \AT{M_1}}
    }{
      \exists\Of{\AT{M_0^v}, \AT{M_1^v}}{\IVal}.\
      \StepStar{M_i}{M_i^v}
      \land
      \Member{\Parens*{\AT{M_0^v},\AT{M_1^v}}}{\mathcal{A}}
    }
  }
  \\[6pt]
  \ValClo{-}&:\PTSObj\to\PTSObj
  \\
  \ADefine{
    \ValClo{\tau}
  }{
    \SetCompr{
      \Parens*{\AT{A},\mathcal{A}}
    }{
      \exists\Of{\AT{A_v}}{\IVal}.\
      \StepStar{A}{A_v}
      \land
      \Member{\Parens*{\AT{A_v},\mathcal{A}}}{\tau}
    }
  }
\end{align*}

In Figure~\ref{fig:monotone-operator}, for an initial candidate type
system $\Of{\sigma}{\PTSObj}$, we define an endomorphism on candidate
type systems $\Of{\TSFun{\sigma}}{\PTSObj\to\PTSObj}$ which extends a
type system with all the non-universe connectives of \ClockCTT{}.

\ifreport%
\begin{sidewaysfigure}
  \begin{minipage}[c]{1.0\linewidth}
    \input{figures/monotone-operator.tex}
  \end{minipage}

  \caption{A monotone operator on candidate type systems.}\label{fig:monotone-operator}
\end{sidewaysfigure}
\else%
\begin{figure*}
  \begin{minipage}[c]{1.0\linewidth}
    \input{figures/monotone-operator-short.tex}
  \end{minipage}

  \caption{A monotone operator on candidate type systems; for the sake
    of space, we elide the interpretations of the standard connectives.}\label{fig:monotone-operator}
\end{figure*}
\fi%

\ifreport%

A few remarks on our style of definition are in order. First, observe that we
have \emph{not} required that $\AT{A}$ be a type in order for
$\AT{\CttLater{\kappa}{A}}$ to be a type: we only require that this premise
obtain \emph{later}. This is crucial for the interaction of the later modality
with the dependent product and function types.

Moreover, we have chosen a negative definition of dependent pair and function
types, based on projections and application rather than on pairing and
abstraction. This choice appears to likewise be forced for the same reason.

Finally, in the type-functionality clauses for dependent pair and
function types, we require the family of relations $\mathcal{B}$ to be
not only functional in $\mathcal{A}$ in the obvious sense, but also in
a ``criss-crossed'' sense: for
$\Member{\Parens*{\AT{M_0}, \AT{M_1}}}{\mathcal{A}}$ we additionally
require
$\Member{\Parens*{\AT{\Subst{M_0}{B}},\mathcal{B}(\AT{M_1})}}{\tau}$
and
$\Member{\Parens*{\AT{\Subst{M_1}{B}},\mathcal{B}(\AT{M_0})}}{\tau}$. Ultimately
this is redundant in case $\mathcal{A}$ is symmetric and $\tau$ is
extensional; however, we found that building these extra instances
into the definition made it simpler to prove that the closure of a
type system is both extensional and CPER-valued under suitable
conditions.
\fi%

\begin{theorem}[\CoqRef{Closure.Clo.monotonicity}]
  For any candidate type system $\Of{\sigma}{\PTSObj}$, the function
  $\Of{\TSFun{\sigma}}{\PTSObj\to\PTSObj}$ is monotone.
\end{theorem}

\begin{proof}
  By case on the type closure clauses above, which are themselves each
  monotone.
\end{proof}

\begin{corollary}[\CoqRef{Closure.Clo.t}, \CoqRef{Closure.Clo.roll}]
  By the Knaster-Tarski theorem, the function $\TSFun{\sigma}$ has a
  least fixed point $\mu\Parens*{\TSFun{\sigma}}$.
\end{corollary}

We will write $\Of{\TSClose{-}}{\PTSObj\to\PTSObj}$ for the operator
that takes $\Of{\sigma}{\PTSObj}$ to the fixed point
$\mu\Parens*{\TSFun{\sigma}}$.

\ifreport%
\begin{lemma}[\CoqRef{Closure.Clo.extensionality}]\label{thm:clo-extensionality}
  For any $\Of{\sigma}{\PTSObj}$ an extensional candidate type system
  which contains only types that evaluate to universes, the closure
  $\TSClose{\sigma}$ is extensional.
\end{lemma}

\begin{proof}
  By the universal property of the closure operator.
\end{proof}

\begin{lemma}[\CoqRef{Closure.Clo.cext\_per}, \CoqRef{Closure.Clo.cext\_computational}]\label{thm:eval-clo-per}
  If the relation $\Of{\mathcal{A}}{\Rel{\ITm{0}}}$ is a PER, then
  $\ValClo{\mathcal{A}}$ is a computational PER.%
\end{lemma}

\begin{proof}
  By the determinacy of evaluation.
\end{proof}

\begin{lemma}[\CoqRef{Closure.Clo.cper\_valued}]\label{thm:clo-cper-valued}
  If $\Of{\sigma}{\PTSObj}$ is CPER-valued, extensional and contains
  only types that evaluate to universes, then its closure
  $\TSClose{\sigma}$ is CPER-valued.
\end{lemma}

\begin{proof}
  By the universal property of the closure operator, using
  Theorem~\ref{thm:later-cartesian}.
\end{proof}

\begin{lemma}[\CoqRef{Closure.Clo.type\_computational}]\label{thm:clo-type-computationality}
  If $\Of{\sigma}{\PTSObj}$ is type-computational, then so is its
  closure $\TSClose{\sigma}$.
\end{lemma}

\begin{proof}
  By the universal property of the closure operator, using
  Theorem~\ref{thm:later-cartesian}.
\end{proof}
\fi%

\ifreport
  \begin{theorem}[\CoqRef{Closure.Clo.monotonicity}]
    For any candidate type system $\Of{\sigma}{\PTSObj}$, the function
    $\Of{\TSFun{\sigma}}{\PTSObj\to\PTSObj}$ is monotone.
  \end{theorem}

  \begin{proof}
    By case on the type closure clauses, which are themselves monotone.
  \end{proof}

  \begin{corollary}[\CoqRef{Closure.Clo.t}, \CoqRef{Closure.Clo.roll}]
    By the Knaster-Tarski theorem, the function $\TSFun{\sigma}$ has a
    least fixed point $\mu\Parens*{\TSFun{\sigma}}$.
  \end{corollary}

  We will write $\Of{\TSClose{-}}{\PTSObj\to\PTSObj}$ for the operator
  that takes $\Of{\sigma}{\PTSObj}$ to the fixed point
  $\mu\Parens*{\TSFun{\sigma}}$.
\fi

\subsection{The full universe hierarchy}

\newcommand\Spine[1]{\nu_{#1}}

The next step in the construction is to build up the universe
hierarchy. Following~\citet{allen:1987:thesis}, we define the
``spine'' of the universe hierarchy as a sequence of type systems
$\Of{\Spine{}}{\PTSObj^\Nat}$ that contains at each level only types
which evaluate to universes:
\begin{align*}
  \Spine{0} &= \bot
  \\
  \Spine{n+1} &=
  \ValClo{
    \SetCompr{
      \Parens*{\AT{\ITyUniv{i}},\mathcal{U}}
    }{
      i\leq n
      \land
      \mathcal{U}\equiv
      \SetCompr{
        \Parens*{\AT{A_0}, \AT{A_1}}
      }{
        \WithTS{\TSClose{\Spine{i}}}{
          \SimType{A_0}{A_1}
        }
      }
    }
  }
\end{align*}

The sequence above is well-defined by complete induction on the
index.
\ifreport%
\begin{lemma}[\CoqRef{Tower.Spine.monotonicity}]\label{thm:spine-monotonicity}
  If $\IsLEQ{i}{j}$, then $\IsApprox{\Spine{i}}{\Spine{j}}$.
\end{lemma}

\begin{proof}
  By induction on $i$.
\end{proof}

\begin{lemma}[\CoqRef{Tower.Spine.extensionality}]\label{thm:spine-extensionality}
  Every spine level $\nu_i$ is extensional in the sense that it is the
  graph of a partial function $\ITm{0}\rightharpoonup{}\Rel{\ITm{0}}$.
\end{lemma}

\begin{proof}
  By case on $i$.
\end{proof}

\begin{lemma}[\CoqRef{Tower.Spine.type\_computational}]\label{thm:spine-type-computationality}
  Every spine level $\nu_i$ is type-computational.
\end{lemma}

\begin{proof}
  By case on $i$.
\end{proof}

\begin{lemma}[\CoqRef{Tower.Spine.cper\_valued}]\label{thm:spine-cper-valued}
  Every spine is valued in CPERs.
\end{lemma}

\begin{proof}
  By induction on $i$, using
  Lemmas~\ref{thm:clo-extensionality},~\ref{thm:clo-type-computationality},~\ref{thm:spine-extensionality}
  and Theorem~\ref{thm:spine-type-computationality}.
\end{proof}

\fi%
We are now equipped to define a new sequence of type systems
which is at each level closed under all the ordinary type formers as
well as smaller universes:
\[
  \Define{\TS_n}{
    \TSClose{\Spine{n}}
  }
\]

\ifreport%
\begin{lemma}[\CoqRef{Tower.monotonicity}]\label{thm:tower-monotonicity}
  If $\IsLEQ{i}{j}$, then $\IsApprox{\TS_i}{\TS_j}$.
\end{lemma}

\begin{proof}
  By the universal property of the closure operator and
  Lemma~\ref{thm:spine-monotonicity}.
\end{proof}

\begin{theorem}[\CoqRef{Tower.extensionality},\CoqRef{Tower.type\_computational},\\ \CoqRef{Tower.cper\_valued}]\label{thm:tower-type-system}
  Each candidate type system $\tau_i$ is in fact a type system.
\end{theorem}

\begin{proof}
  $\tau_i$ is extensional immediately from
  Lemma~\ref{thm:clo-extensionality} and the fact that the spine
  $\Spine{i}$ contains only types that evaluate to universes. It is
  type-computational by Lemmas~\ref{thm:clo-type-computationality}
  and~\ref{thm:spine-type-computationality}. It is CPER-valued by
  Lemmas~\ref{thm:clo-cper-valued} and~\ref{thm:spine-cper-valued}.
\end{proof}
\fi%

Finally, we can capture the entire countable hierarchy in a single
type system $\TSOmega$, which is the join of the entire sequence:
\[
  \Define{\TSOmega}{
    \bigvee_{i:\Nat}\TS_i
  }
\]

When we explain the meaning of judgments, it will always be done with
respect to this maximal type system.

\begin{theorem}[$\TSOmega$ type system]\label{thm:tauomega-type-system}
  The ultimate candidate type system $\TSOmega$ is in fact a type
  system.
\end{theorem}

\subsection{Meaning explanation}\label{sec:meaning-explanation}

In this section, we give a mathematical meaning explanation to the
formal judgments of \ClockCTT{}:

\begin{enumerate}
\item Functional equality of elements
  $\FHypEqMem{\Delta}{\Gamma}{M_0}{M_1}{A}$ means that in clock
  context $\CT{\Delta}$ and variable context $\CT{\Gamma}$,
  $\CT{M_0}$ and $\CT{M_1}$ are equal elements of type $\CT{A}$. This
  form of judgment requires that
  $\CT{\Gamma},\CT{M_0},\CT{M_1},\CT{A}$ mention only clocks from
  $\CT{\Delta}$, and that $\CT{M_0},\CT{M_1},\CT{A}$ mention only
  variables from $\Gamma$.
\item Untyped open conversion $\FOpenConv{\Delta}{\Psi}{M_0}{M_1}$ means
  that $\CT{M_0}$ and $\CT{M_1}$ are Kleene equivalent in all their
  instantiations. This form of judgment requires that
  $\CT{M_0},\CT{M_1}$ mention only clocks from $\CT{\Delta}$ and
  variables from $\CT{\Psi}$.
\end{enumerate}

\paragraph{The meaning of judgments}
We interpret each formal judgment $\mathcal{J}$ as a proposition
$\Of{\SemJdg{\mathcal{J}}}{\Omega}$ in $\ClkTopos$.
\[
  \begin{array}[t]{l}
    \SemJdg{\FHypEqMem{\Delta}{\Gamma}{M_0}{M_1}{A}}
    \triangleq%
    \\[6pt]
    \quad
    \begin{array}[t]{l}
      \forall\Of{\varrho}{\ClkObj^{\Dom{\Delta}}}.
      \\[6pt]
      \quad
      \WithTS{\TSOmega}{
        \IsCtx{
          \Sem{\Gamma}\varrho
        }
      }
      \\[6pt]
      \quad
      \Rightarrow
      \WithTS{\TSOmega}{
        \HypSimType{
          \Sem{\Gamma}\varrho
        }{
          \Sem{A_0}\varrho
        }{
          \Sem{A_1}\varrho
        }
      }
      \\[6pt]
      \quad
      \Rightarrow
      \WithTS{\TSOmega}{
        \HypSimMem{
          \Sem{\Gamma}\varrho
        }{
          \Sem{M_0}\varrho
        }{
          \Sem{M_1}\varrho
        }{
          \Sem{A}\varrho
        }
      }
    \end{array}
    \\\\
    \SemJdg{\FOpenConv{\Delta}{\Psi}{M_0}{M_1}}
    \triangleq%
    \forall\Of{\varrho}{\ClkObj^{\Dom{\Delta}}}.\
    \OpenSq{\Dom{\CT{\Psi}}}{
      \Sem{M_0}\varrho%
    }{
      \Sem{M_1}\varrho%
    }
  \end{array}
\]

Observe that the usual presuppositions of the equality judgment
(context validity and type functionality) are taken as
\emph{assumptions}: the principle can be summarized as ``garbage in,
garbage out''. Dually, we could have chosen to regard them as
consequences, which would lead to a slightly different collection of
validated rules.

\paragraph{Canonicity at base type}
Write $\Of{\mathbf{2}}{\ClkTopos}$ for the boolean object in our
semantic framework which has two global elements
$\Of{\mathbf{2}_0,\mathbf{2}_1}{\mathbf{2}}$. Define an embedding
$\lfloor-\rfloor_{\mathbf{2}}$ from this object into our formal term
language as follows:
\begin{align*}
  \lfloor\mathbf{2}_0\rfloor_{\mathbf{2}} &= \CT{\Tt}
  \\
  \lfloor\mathbf{2}_1\rfloor_{\mathbf{2}} &= \CT{\Ff}
\end{align*}

Now we can state the canonicity theorem for \ClockCTT{}.%

\begin{theorem}[\CoqRef{Canonicity.canonicity}]\label{thm:canonicity}
  For any closed expression $\CT{M}$ such that
  $\SemJdg{\FHypEqMem{\cdot}{\cdot}{M}{M}{\TyBool}}$, there exists some
  $\Member{b}{\mathbf{2}}$ such that
  $\SemJdg{\FOpenConv{\cdot}{\cdot}{M}{\lfloor b\rfloor_{\mathbf{2}}}}$.
\end{theorem}

\begin{corollary}
  The type theory \ClockCTT{} is consistent in the sense that there is
  no inhabitant of $\CT{\TyVoid}$.
\end{corollary}

\ifreport
  Theorem~\ref{thm:canonicity} is not immediately as strong as one would hope,
  but it implies a strong external result.  Unfolding the $\forall\exists$
  statement of Theorem~\ref{thm:canonicity}, it is easy to see that at each
  individual world there \emph{externally} exists a real boolean which has the
  desired property. To see that there is constructively a way to choose such a
  boolean externally (which is not automatically implied by the Kripke-Joyal
  semantics of $\forall\exists$ statements), it suffices to make the following
  observations.

  In what follows, we will write $\FormalTerm$ for the object of formal terms in
  $\ClkTopos$.

  \begin{enumerate}

    \item Writing $\SemSquares{\TyBool}$ for the subobject
      $\SetCompr{\Of{M}{\FormalTerm}}{\SemJdg{\FHypEqMem{\cdot}{\cdot}{M}{M}{\TyBool}}}$,
      Theorem~\ref{thm:canonicity} states the following:
      \[
        \ClkTopos\models
        \forall\Member{M}{\SemSquares{\TyBool}}.\
        \exists\Of{b}{\mathbf{2}}.\
        \SemJdg{\FOpenConv{\cdot}{\cdot}{M}{\lfloor b\rfloor_{\mathbf{2}}}}
      \]

    \item Observe that internally, the boolean $b$ is uniquely determined. This
      follows from the fact that $\lfloor b\rfloor_{\mathbf{2}}$ is a value, and
      from the determinacy of the evaluation relation.

    \item Therefore, we can strengthen the above to the following:
      \[
        \ClkTopos\models
        \forall\Member{M}{\SemSquares{\TyBool}}.\
        \exists!\Of{b}{\mathbf{2}}.\
        \SemJdg{\FOpenConv{\cdot}{\cdot}{M}{\lfloor b\rfloor_{\mathbf{2}}}}
      \]

    \item By the axiom of unique choice (which holds in every topos), the above
      is equivalent to the following:
      \[
        \ClkTopos\models
        \exists\Of{F}{{\mathbf{2}}^{\SemSquares{\TyBool}}}.\
        \forall\Member{M}{\SemSquares{\TyBool}}.\
        \SemJdg{\FOpenConv{\cdot}{\cdot}{M}{\lfloor F(M)\rfloor_{\mathbf{2}}}}
      \]

    \item Unfolding this existential in the Kripke-Joyal semantics, choosing any
      world $\mathbf{U}$, we can exhibit externally a section of the presheaf exponential
      ${\mathbf{2}}^{\SemSquares{\TyBool}}(\mathbf{U})$. Examining the construction of the
      presheaf exponential, this gives us a metatheoretic function to read back,
      from any definable formal term $\CT{M}$ which satisfies the typing judgment,
      the exact metatheoretic boolean it evaluates to.
  \end{enumerate}

  This can be thought of as an \emph{admissible statement} about the topos logic:
  from a formal term $\CT{M}$ and a proof that it is an element of type
  $\CT{\TyBool}$, we can extract an external boolean which has the desired
  property.
\fi

\subsection{Validated rules}

\ifreport%
We have validated the following rules for \ClockCTT{} in our Coq
formalization.
\begin{mathparpagebreakable}
  \inferrule[\CoqRef{Conversion.symm}]{
    \FOpenConv{\Delta}{\Psi}{M_0}{M_1}
  }{
    \FOpenConv{\Delta}{\Psi}{M_1}{M_0}
  }
  \and
  \inferrule[\CoqRef{Conversion.Trans}]{
    \FOpenConv{\Delta}{\Psi}{M_0}{M_1}
    \\
    \FOpenConv{\Delta}{\Psi}{M_1}{M_2}
  }{
    \FOpenConv{\Delta}{\Psi}{M_0}{M_2}
  }
  \and
  \inferrule[\CoqRef{General.weakening}]{
    \FHypEqMem{\Delta}{\Gamma}{M_0}{M_1}{A}
  }{
    \FHypEqMem{\Delta}{\Gamma,x:B}{M_0}{M_1}{A}
  }
  \and
  \inferrule[\CoqRef{General.hypothesis}]{
  }{
    \FHypIsMem{\Delta}{\Gamma,x:\alpha}{x}{\alpha}
  }
  \and
  \inferrule[\CoqRef{General.conv\_mem}]{
    \FHypEqMem{\Delta}{\Gamma}{M_{01}}{M_1}{\alpha}
    \\
    \IsEq{\pi(\CT{\Gamma})}{\CT{\Psi}}
    \\
    \FOpenConv{\Delta}{\Psi}{M_{00}}{M_{01}}
  }{
    \FHypEqMem{\Delta}{\Gamma}{M_{00}}{M_1}{\alpha}
  }
  \and
  \inferrule[\CoqRef{General.conv\_ty}]{
    \FHypEqMem{\Delta}{\Gamma}{M_0}{M_1}{A_1}
    \\
    \IsEq{\pi(\CT{\Gamma})}{\CT{\Psi}}
    \\
    \FOpenConv{\Delta}{\Psi}{A_0}{A_1}
  }{
    \FHypEqMem{\Delta}{\Gamma}{M_0}{M_1}{A_0}
  }
  \and
  \inferrule[\CoqRef{General.eq\_symm}]{
    \FHypEqMem{\Delta}{\Gamma}{M_0}{M_1}{A}
  }{
    \FHypEqMem{\Delta}{\Gamma}{M_1}{M_0}{A}
  }
  \and
  \inferrule[\CoqRef{General.eq\_trans}]{
    \FHypEqMem{\Delta}{\Gamma}{M_1}{M_2}{A}
    \\
    \FHypEqMem{\Delta}{\Gamma}{M_0}{M_1}{A}
  }{
    \FHypEqMem{\Delta}{\Gamma}{M_0}{M_2}{A}
  }
  \and
  \inferrule[\CoqRef{General.replace\_ty}]{
    \FHypEqMem{\Delta}{\Gamma}{A_0}{A_1}{\TyUniv{i}}
    \\
    \FHypEqMem{\Delta}{\Gamma}{M_0}{M_1}{A_0}
  }{
    \FHypEqMem{\Delta}{\Gamma}{M_0}{M_1}{A_1}
  }
  \and
  \inferrule[\CoqRef{General.univ\_formation}]{
    \Parens*{i<j}
  }{
    \FHypIsMem{\Delta}{\Gamma}{\TyUniv{i}}{\TyUniv{j}}
  }
  \and
  \inferrule[\CoqRef{Unit.ax\_equality}]{
  }{
    \FHypIsMem{\Delta}{\Gamma}{\bigstar}{\TyUnit}
  }
  \and
  \inferrule[\CoqRef{Bool.univ\_eq}]{
  }{
    \FHypIsMem{\Delta}{\Gamma}{\TyBool}{\TyUniv{i}}
  }
  \and
  \inferrule[\CoqRef{Bool.tt\_equality}]{
  }{
    \FHypIsMem{\Delta}{\Gamma}{\Tt}{\TyBool}
  }
  \and
  \inferrule[\CoqRef{Bool.ff\_equality}]{
  }{
    \FHypIsMem{\Delta}{\Gamma}{\Ff}{\TyBool}
  }
  \and
  \inferrule[\CoqRef{Prod.univ\_eq}]{
    \FHypEqMem{\Delta}{\Gamma}{A_0}{A_1}{\TyUniv{i}}
    \\
    \FHypEqMem{\Delta}{\Gamma,x:A_0}{B_0}{B_1}{\TyUniv{i}}
  }{
    \FHypEqMem{\Delta}{\Gamma}{
      \DSum{x}{A_0}{B_0}
    }{
      \DSum{x}{A_1}{B_1}
    }{
      \TyUniv{i}
    }
  }
  \and
  \inferrule[\CoqRef{Prod.intro}]{
    \FHypIsMem{\Delta}{\Gamma}{A}{\TyUniv{i}}
    \\
    \FHypIsMem{\Delta}{\Gamma,x:A}{B}{\TyUniv{i}}
    \\\\
    \FHypEqMem{\Delta}{\Gamma}{M_{00}}{M_{10}}{A}
    \\
    \FHypEqMem{\Delta}{\Gamma}{M_{01}}{M_{11}}{[M_{00}/x]B}
  }{
    \FHypEqMem{\Delta}{\Gamma}{
      \Pair{M_{00}}{M_{01}}
    }{
      \Pair{M_{10}}{M_{11}}
    }{
      \DSum{x}{A}{B}
    }
  }
  \and
  \inferrule[\CoqRef{Arr.univ\_eq}]{
    \FHypEqMem{\Delta}{\Gamma}{A_0}{A_1}{\TyUniv{i}}
    \\
    \FHypEqMem{\Delta}{\Gamma,x:A_0}{B_0}{B_1}{\TyUniv{i}}
  }{
    \FHypEqMem{\Delta}{\Gamma}{\DProd{x}{A_0}{B_0}}{\DProd{x}{A_1}{B_1}}{\TyUniv{i}}
  }
  \and
  \inferrule[\CoqRef{Arr.intro}]{
    \FHypIsMem{\Delta}{\Gamma}{A}{\TyUniv{i}}
    \\
    \FHypIsMem{\Delta}{\Gamma,x:A}{B}{\TyUniv{i}}
    \\\\
    \FHypEqMem{\Delta}{\Gamma,x:A}{M_0}{M_1}{B}
  }{
    \FHypEqMem{\Delta}{\Gamma}{\Lam{x}{M_0}}{\Lam{x}{M_1}}{\DProd{x}{A}{B}}
  }
  \and
  \inferrule[\CoqRef{Arr.elim}]{
    \FHypIsMem{\Delta}{\Gamma}{A}{\TyUniv{i}}
    \\
    \FHypIsMem{\Delta}{\Gamma,x:A}{B}{\TyUniv{i}}
    \\\\
    \FHypEqMem{\Delta}{\Gamma}{M_0}{M_1}{\DProd{x}{A}{B}}
    \\
    \FHypEqMem{\Delta}{\Gamma}{N_0}{N_1}{A}
  }{
    \FHypEqMem{\Delta}{\Gamma}{M_0(N_0)}{M_1(N_1)}{[N_0/x]B}
  }
  \and
  \inferrule[\CoqRef{KArr.univ\_eq}]{
    \FHypEqMem{\Delta,k}{\Gamma}{A_0}{A_1}{\TyUniv{i}}
  }{
    \FHypEqMem{\Delta}{\Gamma}{\ClkProd{k}{A_0}}{\ClkProd{k}{A_1}}{\TyUniv{i}}
  }
  \and
  \inferrule[\CoqRef{KArr.intro}]{
    \FHypEqMem{\Delta,k}{\Gamma}{A}{A}{\TyUniv{i}}
    \\
    \FHypEqMem{\Delta,k}{\Gamma}{M_0}{M_1}{A}
  }{
    \FHypEqMem{\Delta}{\Gamma}{\Lam{k}{M_0}}{\Lam{k}{M_1}}{\ClkProd{k}{A}}
  }
  \and
  \inferrule[\CoqRef{KArr.elim}]{
    \FHypEqMem{\Delta,k',k}{\Gamma}{A}{A}{\TyUniv{i}}
    \\
    \FHypEqMem{\Delta,k'}{\Gamma}{M_0}{M_1}{\ClkProd{k}{A}}
  }{
    \FHypEqMem{\Delta,k'}{\Gamma}{M_0(k')}{M_1(k')}{[k'/k]A}
  }
  \and
  \inferrule[\CoqRef{Isect.univ\_eq}]{
    \FHypEqMem{\Delta,k}{\Gamma}{A_0}{A_1}{\TyUniv{i}}
  }{
    \FHypEqMem{\Delta}{\Gamma}{\ClkIsect{k}{A_0}}{\ClkIsect{k}{A_1}}{\TyUniv{i}}
  }
  \and
  \inferrule[\CoqRef{Isect.intro}]{
    \FHypEqMem{\Delta,k}{\Gamma}{M_0}{M_1}{A}
    \\
    \FHypIsMem{\Delta,k}{\Gamma}{A}{\TyUniv{i}}
  }{
    \FHypEqMem{\Delta}{\Gamma}{M_0}{M_1}{\ClkIsect{k}{A}}
  }
  \and
  \inferrule[\CoqRef{Isect.irrelevance}]{
    \FHypIsMem{\Delta}{\Gamma}{A}{\TyUniv{i}}
    \\
    \Parens*{\CT{k}\notin\CT{\Delta}}
  }{
    \FHypEqMem{\Delta}{\Gamma}{A}{\ClkIsect{k}{A}}{\TyUniv{i}}
  }
  \and
  \inferrule[\CoqRef{Isect.preserves\_sigma}]{
    \FHypEqMem{\Delta,k}{\Gamma}{A_0}{A_1}{\TyUniv{i}}
    \\
    \FHypEqMem{\Delta,k}{\Gamma}{B_0}{B_1}{\TyUniv{i}}
  }{
    \FHypEqMem{\Delta}{\Gamma}{
      \ClkIsect{k}{\Parens*{\DSum{x}{A_0}{B_0}}}
    }{
      \DSum{x}{\ClkIsect{k}{A_0}}{\ClkIsect{k}{B_0}}
    }{
      \TyUniv{i}
    }
  }
  \and
  \inferrule[\CoqRef{Later.univ\_eq}]{
    \FHypEqMem{\Delta,k}{\Gamma}{A_0}{A_1}{\CttLater{k}{\TyUniv{i}}}
  }{
    \FHypEqMem{\Delta,k}{\Gamma}{\CttLater{k}{A_0}}{\CttLater{k}{A_1}}{\TyUniv{i}}
  }
  \and
  \inferrule[\CoqRef{Later.intro}]{
    \FHypEqMem{\Delta,k}{\Gamma}{M_0}{M_1}{A}
    \\
    \FHypIsMem{\Delta,k}{\Gamma}{A}{\TyUniv{i}}
  }{
    \FHypEqMem{\Delta,k}{\Gamma}{M_0}{M_1}{\CttLater{k}{A}}
  }
  \and
  \inferrule[\CoqRef{Later.force}]{
    \FHypEqMem{\Delta}{\Gamma}{\ClkIsect{k}{A_0}}{\ClkIsect{k}{A_1}}{\TyUniv{i}}
  }{
    \FHypEqMem{\Delta}{\Gamma}{\ClkIsect{k}{\CttLater{k}{A_0}}}{\ClkIsect{k}{A_1}}{\TyUniv{i}}
  }
  \and
  \inferrule[\CoqRef{Later.preserves\_pi}]{
    \FHypEqMem{\Delta}{\Gamma}{A_0}{A_1}{\TyUniv{i}}
    \\
    \FHypEqMem{\Delta}{\Gamma,x:A}{B_0}{B_1}{\CttLater{k}{\TyUniv{i}}}
  }{
    \FHypEqMem{\Delta}{\Gamma}{\CttLater{\kappa}{\Parens*{\DProd{x}{A_0}{B_0}}}}{\DProd{x}{\CttLater{k}{A_1}}{\CttLater{k}{B_1}}}{\TyUniv{i}}
  }
  \and
  \inferrule[\CoqRef{Later.preserves\_sigma}]{
    \FHypEqMem{\Delta}{\Gamma}{A_0}{A_1}{\TyUniv{i}}
    \\
    \FHypEqMem{\Delta}{\Gamma,x:A}{B_0}{B_1}{\CttLater{k}{\TyUniv{i}}}
  }{
    \FHypEqMem{\Delta}{\Gamma}{\CttLater{\kappa}{\Parens*{\DSum{x}{A_0}{B_0}}}}{\DSum{x}{\CttLater{k}{A_1}}{\CttLater{k}{B_1}}}{\TyUniv{i}}
  }
  \and
  \inferrule[\CoqRef{Later.induction}]{
    \FHypEqMem{\Delta,k}{\Gamma,x:\CttLater{k}{A}}{M_0}{M_1}{A}
  }{
    \FHypEqMem{\Delta,k}{\Gamma}{\Fix{x}{M_0}}{\Fix{x}{M_1}}{A}
  }
\end{mathparpagebreakable}

\else%
In Figure~\ref{fig:validated-rules}, we present a small selection of
the rules which we have validated in our Coq formalization of
\ClockCTT{}; we omit most of the standard rules of ordinary
extensional type theory, which are also valid in our semantics.
\begin{figure*}
  \input{figures/validated-rules-short}
  \caption{A selection of the rules which we have proved correct for our
    type theory.}\label{fig:validated-rules}
\end{figure*}
\fi%

\subsection{Examples: revisiting streams}\label{sec:revisiting-streams}

Using these rules, we can derive some typing lemmas for guarded
streams and coinductive sequences of bits.
\begin{align*}
  \ADefine{
    \CT{\Stream}
  }{
    \CT{
      \Lam{k}{
        \Fix{A}{
            \TyBool
          \times
          \CttLater{k}{A}
        }
      }
    }
  }
  \\
  \ADefine{
    \CT{\Sequence}
  }{
    \CT{
      \ClkIsect{k}{
        \Stream\ k
      }
    }
  }
  \\
  \ADefine{
    \CT{\mathtt{ones}}
  }{
    \CT{
      \Fix{x}{
        \Pair{\Tt}{x}
      }
    }
  }
\end{align*}

\begin{mathparpagebreakable}
  \inferrule[\CoqRef{Examples.BitStream\_wf}]{}{
    \FHypIsMem{\Delta}{\Gamma}{\Stream}{\ClkProd{k}{\TyUniv{i}}}
  }
  \and
  \inferrule[\CoqRef{Examples.BitSeq\_wf}]{}{
    \FHypIsMem{\Delta}{\Gamma}{\Sequence}{\TyUniv{i}}
  }
  \and
  \inferrule[\CoqRef{Examples.BitStream\_unfold}]{}{
    \FHypEqMem{\Delta,k}{\Gamma}{\Stream\ k}{\TyBool\times\CttLater{k}{\Stream\ k}}{\TyUniv{i}}
  }
  \and
  \inferrule[\CoqRef{Examples.BitSeq\_unfold}]{}{
    \FHypEqMem{\Delta}{\Gamma}{\Sequence}{\TyBool\times\Sequence}{\TyUniv{i}}
  }
  \and
  \inferrule[\CoqRef{Examples.Ones\_wf\_guarded}]{}{
    \FHypIsMem{\Delta,k}{\Gamma}{\mathtt{ones}}{\Stream\ k}
  }
  \and
  \inferrule[\CoqRef{Examples.Ones\_wf\_infinite}]{}{
    \FHypIsMem{\Delta}{\Gamma}{\mathtt{ones}}{\Sequence}
  }
\end{mathparpagebreakable}


\section{Survey of Related Work}

\subsection{Guarded Dependent Type Theory}

The standard model of guarded recursion without clocks is the
\emph{topos of trees} ${\Psh{\omega}}$, the presheaves on the poset of
natural numbers regarded as a
category~\citep{birkedal-mogelberg-schwinghammer-stovring:2011}. This
topos can be regarded as a denotational model for a variant of
Martin-L\"of's extensional type theory equipped with the
$\ObjLater{}{}$ modality. By indexing this topos over a category of
clock contexts ${\Delta}$, it is possible to develop a model of
extensional type theory with clock quantification called
\GDTT~\citep{bizjak-grathwohl-grathwohl-clouston-mogelberg-birkedal:2016,bizjak-mogelberg:2015}. In
order to justify a crucial \emph{clock irrelevance} principle, it is
necessary to index universes in clock contexts, i.e.\
$\mathcal{U}_\Delta$.

In the dependent setting, some difficulties arise when devising a
\emph{syntax} for the semantic type theory of this indexed
category. In order to make sense of the ``delayed application''
operator $\circledast$ in the context of dependent function types, it
was necessary to introduce a notion of \emph{delayed substitution}
$\IsEq{\xi}{[\overrightarrow{x\leftarrow e}]}$ which pervades the term
language, introducing term formers like ${\triangleright^k\xi.A}$
and ${\mathsf{next}^k\xi.e}$. On the bright side, delayed
application can be defined in terms of delayed substitution.

However, the equational theory for delayed substitutions is fairly
sophisticated, and an operational (computational) interpretation of
\GDTT{} has not yet been proposed at the time this article was
written; as such, a canonicity theorem for this system is still
forthcoming.

\subsection{Orthogonality and clock irrelevance}

In a more recent development~\citep{bizjak-mogelberg:2017}, a
denotational model of \GDTT{} has been developed that differs from
that of \citet{bizjak-mogelberg:2015} in a few crucial ways.

\paragraph{Unified base category}

The fibered topos presentation of the \citet{bizjak-mogelberg:2015}
work has been replaced with a presheaf topos over a single unified
base category, discovered independently from the unified base category
which we introduce in Section~\ref{sec:semantic-universe}. Taking
presheaves over this unified base category simplifies the model
significantly, and also makes available the standard solution to the
substitution coherence problem for (denotational) presheaf models of
dependent type theory.\footnote{This is to use an alternative
  construction of the slice categories $\Slice{\Psh{\mathbb{C}}}{X}$,
  as the presheaves on the total category of $X$.}

The proposed base category of \citet{bizjak-mogelberg:2017} differs
from ours mainly in that they allow empty worlds, whereas we restrict
our base category to those worlds which contain at least a single
clock.

\paragraph{Orthogonality}

\citeauthor{bizjak-mogelberg:2017} define a presheaf of clocks
$\mathcal{C}$ which is the same as our object of clocks $\ClkObj$
which we introduce in Section~\ref{sec:semantic-universe}; then, the
clock quantifier is represented in the internal language of their
presheaf topos as a dependent product over $\mathcal{C}$, i.e.\
$\prod_{x:\mathcal{C}}A(x)$.

Defined in this way, the clock quantifier cannot be a priori
parametric with respect to clocks / time objects; therefore, in order
to validate the clock irrelevance axiom, the authors have identified
an orthogonality condition on objects, which in essence closes
the internal language of the presheaf topos under just those types
which are compatible with the irrelevance principle for the clock
quantifier.

Unfortunately, the subtopos of time-orthogonal objects does not
contain the standard Hofmann-Streicher universes, because universes
necessarily classify types that depend on clocks in an essential
way. In order to resolve this problem, the standard presheaf-theoretic
universe $\mathcal{U}$ is replaced with a family of universes
$\mathcal{U}_\Delta$ for each clock context $\Delta$; each universe
$\mathcal{U}_\Delta$ classifies the types which may depend only on the
clocks in $\Delta$.

\paragraph{Discussion}

Temporarily abstracting away from the differences between a
denotational account of \GDTT{} and our operational account of type
theory, we can briefly summarize the difference between our approaches
to clock quantification and irrelevance.

The approach of \citet{bizjak-mogelberg:2017} is in essence to define clock
quantification as a dependent (cartesian) product, and then restrict the
available semantic constructions to precisely those which treat clocks
parametrically; then, within this subcategory, the clock quantifier can itself
be regarded as a parametric quantifier (because all counterexamples have been
muted).

Our approach is instead to define clock quantifiers which \emph{intrinsically}
behave in the desired way, rather than starting with only a proof-relevant
quantifier and ruling out observations of its non-parametric character using a
global orthogonality condition. To that end, we have defined two separate clock
quantifiers which decompose the two disjoint uses of $\forall\kappa$
from \GDTT{}:

\begin{enumerate}

  \item A parametric quantifier $\CT{\ClkIsect{k}{A}}$ for expressing that a
    program exhibits a behavior relative to all clocks simultaneously.
    Semantically, this is an intersection, though we expect that a more refined
    perspective will arise as we explore other kinds of model where the
    intersection may not be available.

  \item A non-parametric quantifier $\CT{\ClkProd{k}{A}}$ for internalizing a
    family of objects which varies in a clock; semantically this is the
    cartesian product of a clock-indexed family of types (i.e.\ the right
    adjoint to weakening). \emph{A priori} there is no need for this quantifier
    to behave parametrically, as this is neither demanded nor desired when
    forming families of objects.

\end{enumerate}

In this way, we have managed to avoid imposing any global orthogonality
condition on the objects of our semantic model, leading to a smoother treatment
of universes that avoids indexing in clock contexts.

\subsection{Guarded Cubical Type Theory}

One way to achieve a decidable typing judgment for \GDTT{} is to adopt
an intensional equality, and replace various \emph{judgmental}
principles with propositional axioms (such as the unfolding rule for
$\mathtt{fix}$, as well as several other principles having to do with
identity types which are validated in extensional \GDTT{}). However,
such axioms are disruptive to the computational character of type
theory.

A more refined and well-behaved version of this idea can be found in
Guarded Cubical Type Theory (\GCTT{})
by~\citet{birkedal-bizjak-clouston-grathwohl-spitters-vezzosi:2016},
where $\mathtt{fix}$ is actually exhibited as a
higher-dimensional term, a \emph{line} or \emph{path} between a formal
fixed point and its one-step unfolding.

\GCTT{} currently supports only a single clock, but it is plausible
that it could be extended in the same way as \GDTT{} extends the
internal type theory of the topos of trees.
Although \GCTT{} does not at the time of writing have a decidable typing
result, nor a strong normalization theorem, we are confident that
these can be achieved in the future in light of the intensional
judgmental equality and the restricted unfoldings of fixed points.

\subsection{Clocked Type Theory}
Recently, an alternative to \GDTT{} called \emph{Clocked Type Theory}
(\CloTT) has been proposed, which enjoys a computational
interpretation with a canonicity
result~\citep{bahr-grathwohl-mogelberg:2017}; it is plausible that
Clocked Type Theory shall have a decidable typing relation.
Notably, Clocked Type Theory does not validate any clock irrelevance
rule; the authors propose to address this in a \emph{cubical} version
of \CloTT{} by adding a special path axiom which realizes this
principle, by analogy with the technique used in \GCTT{} to account
for restricted unfoldings of fixed points. In the presence of this
axiom, canonicity for \CloTT{} can still be made to hold in the
context which contains only a single clock.

\paragraph{Discussion}

Clocked Type Theory looks like a promising path toward a well-behaved intrinsic
account of guarded recursion with clocks. In the present paper, our efforts
have been focused exclusively on developing the behavioral account of guarded
type theory in the style of Martin-L\"of's meaning explanation, in which
programs can be regarded as existing separately from their types; here, general
recursive programs can be written and shown to be (causal, productive, total)
in a semantic sense, using the type theory as a program logic.

We perceive, however, that virtue lies in pursuing the intrinsic path,
especially as far as implementability are concerned. The calculus developed
in~\citet{bahr-grathwohl-mogelberg:2017} (and more recently, the ideas
contained in~\citet{clouston-mannaa-mogelberg-pitts-spitters:2018}) are likely
to provide the basis for a syntactic account of guarded recursion which is
sound for our model, but closer to implementation.

\subsection{Sized Types and size quantifiers}

Our decomposition of the quantifier $\forall\kappa$ from \GDTT{} into a
parametric part $\CT{\ClkIsect{k}{A}}$ and a non-parametric part
$\CT{\ClkProd{k}{A}}$ mirrors the state of affairs in the literature on sized
types, which is another account of type-based guarded
recursion~\citep{abel-vezzosi-winterhalter:2017}.

\section{Perspective and Future Work}\label{sec:conclusion}

We have developed and formalized a computational account of guarded
dependent type theory with clocks, enjoying several desirable
characteristics not found together in other existing models:
computational canonicity, clock irrelevance and ordinary universes.
We have made the following contributions toward a simpler, more
computational account of guarded dependent type theory:

\paragraph{Implementation, proof theory, and syntax}

We have not yet tackled the project of developing an ergonomic proof theory for
\ClockCTT{} which can be used to interact with the semantics presented here.
The natural deduction style rules which we have given here are, while
convenient for paper presentations, not what one would use in a serious
implementation. To build a proof theory for \ClockCTT{}, we must negotiate new
forms of judgment with decidable presupposition.

Therefore, while we have indeed developed a programming language for guarded
type theory with clocks that omits explicit syntax for delayed substitutions,
this should be understood in terms of the conceptual order of semantics and
proof theory which is endemic in computational type theory. In particular,
while our programming language and type theory has no need for such a
construct, in a proof language for \ClockCTT{} it would be necessary to account
for the syntactic structure of the later modality's elimination; we anticipate
that ideas from~\citet{bahr-grathwohl-mogelberg:2017}
and~\citet{clouston-mannaa-mogelberg-pitts-spitters:2018} will be highly
relevant.

\paragraph{Application to denotational semantics}
In the future, we are interested in extending our work to a
denotational account of guarded dependent type theory with clocks
which uses the ordinary non-indexed presheaf-topos-theoretic
universe. While our results have been developed in the context of
computational type theory and operational semantics, we believe that
the insight which enabled us to combine clock irrelevance with
ordinary universes is more broadly applicable.


\appendix

\section{Semantic Universe}

In this appendix, we give some further details of the semantic
universe $\ClkTopos$.

\subsection{Internal Logic and Kripke-Joyal Semantics}

Using a tool called Kripke-Joyal semantics (a topos-theoretic
generalization of Beth/Kripke-forcing) it is possible to interpret
statements in the internal language of $\ClkTopos$ into ordinary,
external mathematical language.
We will write forcing clauses $\Forces{\mathbf{U}}{\phi(\alpha)}$ meaning
that at world $\Of{\mathbf{U}}{\CLK}$, the predicate ${\phi}$ holds of the
element $\Of{\alpha}{X(\mathbf{U})}$. The forcing clauses for the
predicates of our internal logic are summarized in
Figure~\ref{fig:internal-logic-forcing-clauses}.

\begin{figure*}
  \begin{minipage}[c]{1.0\textwidth}
  \begin{gather*}
    \JdgDecl{\Forces{\mathbf{U}}{\phi(\alpha)}}{
      \IsSubobject{\phi}{X}{\ClkTopos},
      \Member{\alpha}{X(\mathbf{U})}
    }
    \\[6pt]
    \begin{aligned}
      \AIsEq{
        \Forces{\mathbf{U}}{\phi(\alpha)\lor\psi(\alpha)}
      }{
        \Forces{\mathbf{U}}{\phi(\alpha)}
        \lor
        \Forces{\mathbf{U}}{\psi(\alpha)}
      }
      \\
      \AIsEq{
        \Forces{\mathbf{U}}{\phi(\alpha)\land\psi(\alpha)}
      }{
        \Forces{\mathbf{U}}{\phi(\alpha)}\land\Forces{\mathbf{U}}{\psi(\alpha)}
      }
      \\
      \AIsEq{
        \Forces{\mathbf{U}}{\phi(\alpha)\Rightarrow\psi(\alpha)}
      }{
        \forall\Of{\rho}{\mathbf{V}\to\mathbf{U}}.\
        \Forces{\mathbf{V}}{\phi(\rho^*\alpha)}
        \Rightarrow
        \Forces{\mathbf{V}}{\psi(\rho^*\alpha)}
      }
      \\
      \AIsEq{
        \Forces{\mathbf{U}}{\forall y:Y.\ \phi(\alpha,y)}
      }{
        \forall\Of{\rho}{\mathbf{V}\to\mathbf{U}}.\
        \forall\Member{\beta}{Y(\mathbf{V})}.\
        \Forces{\mathbf{V}}{\phi(\rho^*\alpha,\beta)}
      }
      \\
      \AIsEq{
        \Forces{\mathbf{U}}{\exists y:Y.\ \phi(\alpha,y)}
      }{
        \exists\Member{\beta}{Y(\mathbf{U})}.\
        \Forces{\mathbf{U}}{\phi(\alpha,\beta)}
      }
      \\
      \AIsEq{
        \Forces{\mathbf{U}}{\Later{\kappa}{\phi(\alpha)}}
      }{
        \begin{cases}
          \top &{\normalcolor\textbf{if}\ \ \IsEq{\partial_U(\kappa)}{0}}
          \\
          \Forces{
            \mathbf{U}[\kappa\mapsto n]
          }{
            \phi({[\kappa\pluseq1]}^*\alpha)
          } &{\normalcolor\textbf{if} \ \ \IsEq{\partial_U(\kappa)}{n+1}}
        \end{cases}
      }
    \end{aligned}
  \end{gather*}
\end{minipage}

  \caption{Forcing clauses for the internal logic of $\ClkTopos$.}\label{fig:internal-logic-forcing-clauses}
\end{figure*}

It will simplify many of our proofs to formalize some proof techniques
for establishing that a formula headed by multiple universal
quantifiers is valid in $\ClkTopos$, i.e.\ true at each world.

\begin{lemma}\label{lem:kj-forall}
  To show that a formula $\forall\Of{y}{Y}.\ \phi(\alpha,y)$ is true
  for all worlds $\mathbf{U}$ and elements
  $\Member{\alpha}{X\Parens{\mathbf{U}}}$ in $\ClkTopos$, it suffices to
  establish externally the following statement:
  \[
    \forall\Of{\mathbf{U}}{\CLK}.\
    \forall\Member{\alpha}{X\Parens{\mathbf{U}}}.\
    \forall\Member{\beta}{Y\Parens{\mathbf{U}}}.\
    \Forces{\mathbf{U}}{\phi(\alpha,\beta)}
  \]
\end{lemma}
\begin{proof}
  Fixing a world $\mathbf{U}$ and an element
  $\Member{\alpha}{X\Parens{\mathbf{U}}}$, our original formula unfolds to
  the following in the Kripke-Joyal semantics:
  \[
    \forall\Of{\mathbf{V}}{\CLK}.\
    \forall\Of{\rho}{\mathbf{V}\to\mathbf{U}}.\
    \forall\Member{\beta}{Y\Parens{\mathbf{V}}}.\
    \Forces{\mathbf{V}}{\phi(\rho^*\alpha,\beta)}
  \]

  Fix $\Of{\mathbf{V}}{\CLK}$, $\Of{\rho}{\mathbf{V}\to\mathbf{U}}$ and
  $\Member{\beta}{Y\Parens{\mathbf{V}}}$. By instantiating our assumption
  with $\mathbf{V}$, $\rho^*\alpha$ and $\beta$, we have
  $\Forces{\mathbf{V}}{\phi(\rho^*\alpha,\beta)}$.
\end{proof}

\begin{lemma}\label{lem:kj-multi-forall}
  To show that a formula
  $\forall\overrightarrow{y_i:Y_i}.\
  \phi\Parens{\overrightarrow{y_i},\alpha}$ is true at all worlds
  $\mathbf{U}$ and elements $\Member{\alpha}{X\Parens{\mathbf{U}}}$, it suffices
  to establish the following external statement:
  \[
    \forall\Of{\mathbf{U}}{\CLK}.\
    \forall\overrightarrow{\Member{y_i}{y_i\Parens{\mathbf{U}}}}.\
    \Forces{\mathbf{U}}{\phi\Parens{\overrightarrow{y_i},\alpha}}
  \]
\end{lemma}
\begin{proof}
  Observe that our original formula is logically equivalent to the
  following one with only a single quantifier:
  \[
    \textstyle%
    \forall\Of{y}{\prod_i Y_i}.\
    \phi\Parens{
      \overrightarrow{
        \pi_i(y)
      },
      \alpha
    }
  \]

  Therefore, our goal follows from Lemma~\ref{lem:kj-forall}.
\end{proof}

\begin{lemma}\label{lem:kj-forall-implies}
  To show that a formula
  $\forall\overrightarrow{y_i:Y_i}.\
  \overrightarrow{\phi_j\Parens{\overrightarrow{y_i},\alpha}}\Rightarrow\psi\Parens{\overrightarrow{y_i},\alpha}$
  is true at all worlds $\mathbf{U}$ and elements
  $\Member{\alpha}{X\Parens{\mathbf{U}}}$, it suffices to establish the
  following external statement:
  \[
    \forall\Of{\mathbf{U}}{\CLK}.\
    \forall\overrightarrow{\Member{y_i}{Y_i\Parens{\mathbf{U}}}}.\
    \overrightarrow{
      \Forces{\mathbf{U}}{\phi\Parens{\overrightarrow{y_i},\alpha}}
    }
    \Rightarrow
    \Forces{\mathbf{U}}{\psi\Parens{\overrightarrow{y_i},\alpha}}
  \]
\end{lemma}
\begin{proof}
  Observe that any implication $\phi\Rightarrow\psi$ in the internal
  logic can be equivalently written as a universal quantification over
  a subobject comprehension
  $\forall\Of{x}{\SetCompr{x:\mathbf{1}}{\phi}}.\ \psi$ Therefore, our
  lemma follows from Lemma~\ref{lem:kj-forall}.
\end{proof}

\subsection{Semantic Lemmas}
\begin{theorem}[Local clock]\label{thm:local-clock}
  The formula $\exists\Of{\kappa}{\ClkObj}.\ \top$ is true in the
  internal logic of $\ClkTopos$.
\end{theorem}
\begin{proof}
  It suffices to validate this formula at each world $\mathbf{U}$,
  i.e.\ to establish
  $\Forces{\mathbf{U}}{\exists\Of{\kappa}{\ClkObj}.\ \top}$, which is to
  say (externally) that
  $\exists\Of{\kappa}{\ClkObj\Parens{\mathbf{U}}}.\ \top$. This reduces to
  showing that the hom set $U\to\bullet^1$ in $\FINPlus$ is non-empty,
  which is true because $\FINPlus$ is a category of \emph{non-empty} finite
  products.
\end{proof}

Note that Theorem~\ref{thm:local-clock} does \emph{not} entail the
existence of a global element of $\ClkObj$ (i.e.\ a morphism
$\mathbf{1}\to\ClkObj$). In our development, we have no need for a
global clock; we only require that a clock ``merely exists'' according
to the existential quantifier of the topos logic.

\begin{corollary}[Clock irrelevance]\label{cor:clock-irrelevance}
  The formula
  $\forall\Of{\phi}{\Omega}.\
  \IsEq{\phi}{\forall\Of{\kappa}{\ClkObj}.\ \phi}$ holds in the
  internal logic.
\end{corollary}
\begin{proof}
  We will reason internally: fix $\Of{\phi}{\Omega}$. By propositional
  extensionality we need to show that
  $\phi\Rightarrow\ClkForall{\kappa}{\phi}$ and
  $\ClkForall{\kappa}{\phi}\Rightarrow\phi$. The first direction is
  trivial; for the second direction, observe that from
  Theorem~\ref{thm:local-clock}, using the elimination rule for the
  existential quantifier, we may fix a clock $\Of{\kappa_0}{\ClkObj}$;
  using this clock, by the elimination rule of the universal
  quantifier, we have our goal $\phi$.
\end{proof}

\begin{theorem}\label{thm:forcing}
  We can delete a later modality from under an appropriate
  quantification, in the sense that the following formula is true in
  the internal logic:
  \[
    \forall\Of{\phi}{\Omega^\ClkObj}.\
    \Parens*{\ClkForall{\kappa}{\Later{\kappa}{\phi(\kappa)}}}
    \Rightarrow
    \ClkForall{\kappa}{\phi(\kappa)}
  \]
\end{theorem}
\begin{proof}
  We will establish this principle using the Kripke-Joyal semantics;
  using Lemma~\ref{lem:kj-forall-implies}, we fix a world $\mathbf{U}$ and
  a predicate $\Member{\phi}{{\Omega^\ClkObj}\Parens{\mathbf{U}}}$ such
  that
  $\Forces{\mathbf{U}}{\ClkForall{\kappa}{\Later{\kappa}{\phi(\kappa)}}}$,
  to show $\Forces{\mathbf{U}}{\ClkForall{\kappa}{\phi(\kappa)}}$.

  Observe that our goal is equivalent to the following external
  statement, writing $\pi_1[n],\pi_2[n]$ for the projections of
  $\mathbf{U}$ and $(1,[n])$, respectively, from the extended world
  $(U+1,[\partial_U,n])$:\footnote{This is a special case of the
    ``alternative'' forcing clause (vi$'$) for the universal
    quantifier in Kripke-Joyal semantics, as given in \citet[p.\
    305]{maclane-moerdijk:1992}.}
  \begin{equation}\label{eq:thm:forcing:goal}
    \forall\Member{n}{\omega}.\
    \Forces{
      (U+1,[\partial_U,n])
    }{
      \Parens{\pi_1[n]}^*\phi(\pi_2[n])
    }
    \tag{G1}
  \end{equation}

  In the same way, our premise can be rewritten as follows:
  \begin{equation}\label{eq:thm:forcing:premise}
    \forall\Member{n}{\omega}.\
    \Forces{
      (U+1,[\partial_U,n])
    }{
      \Later{\Parens{\pi_2[n]}}{\Parens{\pi_1[n]}^*\phi(\pi_2[n])}
    }
    \tag{H1}
  \end{equation}

  To establish~(\ref{eq:thm:forcing:goal}), fix $\Member{m}{\omega}$;
  our goal now becomes:
  \begin{equation}\label{eq:thm:forcing:goal:massaged}
    \Forces{(U+1,[\partial_U,m])}{\Parens{\pi_1[m]}^*\phi(\pi_2[m])}
    \tag{G2}
  \end{equation}

  Next instantiate~(\ref{eq:thm:forcing:premise}) with
  $\IsEq{n}{m+1}$, yielding:
  \begin{equation}\label{eq:thm:forcing:premise:massaged}
    \Forces{
      (U+1,[\partial_U,m+1])
    }{
      \Later{\Parens{\pi_2[m+1]}}{\Parens{\pi_1[m+1]}^*\phi(\pi_2[m+1])}
    }
    \tag{H2}
  \end{equation}

  Using the forcing clause for the later modality, we see that
  (\ref{eq:thm:forcing:premise:massaged}) is actually the same as the
  goal~(\ref{eq:thm:forcing:goal:massaged}).
\end{proof}

\begin{theorem}\label{thm:later-unit}
  We have the following unit law in the internal logic:
  \[
    \ClkForall{\kappa}{
      \forall\Of{\phi}{\Omega}.\
      \phi\Rightarrow\Later{\kappa}{\phi}
      }
  \]
\end{theorem}
\begin{proof}
  By Lemma~\ref{lem:kj-forall-implies}, it suffices to fix a world
  $\mathbf{U}$ and elements $\Member{\kappa}{\ClkObj\Parens{\mathbf{U}}}$,
  $\Member{\phi}{\Omega\Parens{\mathbf{U}}}$ such that
  $\Forces{\mathbf{U}}{\phi}$. We need to show that
  $\Forces{\mathbf{U}}{\Later{\kappa}{\phi}}$. Proceed by case on
  $\partial_U(\kappa)$:
  \begin{proofcases}
  \item[$\IsEq{\partial_U(\kappa)}{0}$]
    Immediate.

  \item[$\IsEq{\partial_U(\kappa)}{n+1}$] We need to show that
    $\Forces{\mathbf{U}[\kappa\mapsto
      n]}{\Squares{\kappa\pluseq1}^*\phi}$; this follows by reindexing
    our assumption that $\Forces{\mathbf{U}}{\phi}$.
  \end{proofcases}
\end{proof}

\begin{theorem}\label{thm:later-cartesian}
  The later modality commutes with conjunction:
  \[
    \ClkForall{\kappa}{
      \forall\Of{\phi,\psi}{\Omega}.\
      \IsEq{
        \Later{\kappa}{
          \Parens*{\phi\land\psi}
        }
      }{
        \Parens*{
          \Later{\kappa}{\phi}
          \land
          \Later{\kappa}{\psi}
        }
      }
    }
  \]
\end{theorem}
\begin{proof}
  It suffices to prove that each direction of this quantified equation
  is valid at all worlds:
  \begin{gather}
    \ClkForall{\kappa}{
      \forall\Of{\phi,\psi}{\Omega}.\
      \Later{\kappa}{
        \Parens*{\phi\land\psi}
      }
      \Rightarrow
      \Parens*{
        \Later{\kappa}{\phi}
        \land
        \Later{\kappa}{\psi}
      }
    }
    \tag{$\Rightarrow$}
    \\
    \ClkForall{\kappa}{
      \forall\Of{\phi,\psi}{\Omega}.\
      \Parens*{
        \Later{\kappa}{\phi}
        \land
        \Later{\kappa}{\psi}
      }
      \Rightarrow
      \Later{\kappa}{
        \Parens*{\phi\land\psi}
      }
    }
    \tag{$\Leftarrow$}
  \end{gather}

  ($\Rightarrow$) Using Lemma~\ref{lem:kj-forall-implies}, we fix a
  world $\mathbf{U}$ and elements $\Member{\kappa}{\ClkObj\Parens{\mathbf{U}}}$,
  $\Member{\phi,\psi}{\Omega\Parens{\mathbf{U}}}$ such that
  $\Forces{\mathbf{U}}{\Later{\kappa}{\Parens{\phi\land\psi}}}$. We need to
  show that
  $\Forces{\mathbf{U}}{\Later{\kappa}{\phi}\land\Later{\kappa}{\psi}}$. Proceed by
  case on $\partial_U(\kappa)$:
  \begin{proofcases}
  \item[$\IsEq{\partial_U(\kappa)}{0}$] Immediate.
  \item[$\IsEq{\partial_U(\kappa)}{n+1}$] Then our assumption is equal
    to
    $\Forces{\mathbf{U}[\kappa\mapsto
      n]}{\Squares{\kappa\pluseq1}^*\phi\land\Squares{\kappa\pluseq1}^*\psi}$,
    which is exactly the same as our goal.
  \end{proofcases}

  ($\Leftarrow$) This direction is analogous.
\end{proof}

\begin{corollary}\label{cor:later-monotone}
  The later modality is monotonic:
  \[
    \ClkForall{\kappa}{
      \forall\Of{\phi,\psi}{\Omega}.\
      \Parens{\phi\Rightarrow\psi}
      \Rightarrow
      \Later{\kappa}{\phi}
      \Rightarrow
      \Later{\kappa}{\psi}
    }
  \]
\end{corollary}
\begin{proof}
  This is a well-known corollary of Theorem~\ref{thm:later-cartesian},
  following for purely algebraic reasons. Reasoning internally, fix
  $\Of{\kappa}{\ClkObj}$ and $\Of{\phi,\psi}{\Omega}$ such that
  $\phi\Rightarrow\psi$ and $\Later{\kappa}{\phi}$; we need to show
  $\Later{\kappa}{\psi}$.

  First, observe that
  $\IsEq{\Parens{\Later{\kappa}{\phi}\land\Later{\kappa}{\psi}}}{\Later{\kappa}{\phi}}$. To
  show that this is the case, by Theorem~\ref{thm:later-cartesian} it
  suffices to show that
  $\IsEq{\Later{\kappa}{\Parens{\phi\land\psi}}}{\Later{\kappa}{\phi}}$. This
  holds, because $\IsEq{\phi\land\psi}{\phi}$:
  $\phi\land\psi\Rightarrow\phi$ is trivial, and
  $\phi\Rightarrow\phi\land\psi$ follows from our assumption
  $\phi\Rightarrow\psi$.

  Returning to our main goal $\Later{\kappa}{\psi}$, using the above,
  we may replace our assumption $\Later{\kappa}{\phi}$ with
  $\Later{\kappa}{\phi}\land\Later{\kappa}{\psi}$, whence we have
  immediately $\Later{\kappa}{\psi}$.
\end{proof}

\begin{theorem}\label{thm:later-commutes-with-implication}
  The later modality commutes with implication:
  \[
    \ClkForall{\kappa}{
      \forall\Of{\phi,\psi}{\Omega}.\
      \IsEq{
        \Later{\kappa}{
          \Parens*{\phi\Rightarrow\psi}
        }
      }{
        \Parens*{
          \Later{\kappa}{\phi}
          \Rightarrow
          \Later{\kappa}{\psi}
        }
      }
    }
  \]
\end{theorem}

\begin{proof}
  As in Theorem~\ref{thm:later-cartesian}, it will be simplest to show
  that each direction of the quantified equation is valid at all
  worlds:
  \begin{gather}
    \ClkForall{\kappa}{
      \forall\Of{\phi,\psi}{\Omega}.\
      \Later{\kappa}{
        \Parens*{\phi\Rightarrow\psi}
      }
      \Rightarrow
      \Parens*{
        \Later{\kappa}{\phi}
        \Rightarrow
        \Later{\kappa}{\psi}
      }
    }
    \tag{$\Rightarrow$}
    \\
    \ClkForall{\kappa}{
      \forall\Of{\phi,\psi}{\Omega}.\
      \Parens*{
        \Later{\kappa}{\phi}
        \Rightarrow
        \Later{\kappa}{\psi}
      }
      \Rightarrow
      \Later{\kappa}{
        \Parens*{\phi\Rightarrow\psi}
      }
    }
    \tag{$\Leftarrow$}
  \end{gather}

  ($\Rightarrow$) We will reason algebraically:
  \begin{align*}
    &\Later{\kappa}{
      \Parens*{\phi\Rightarrow\psi}
    }
    \Rightarrow
    \Parens*{
      \Later{\kappa}{\phi}
      \Rightarrow
      \Later{\kappa}{\psi}
    }
    \\
    &\equiv
    \Later{\kappa}{
      \Parens*{\phi\Rightarrow\psi}
    }
    \land
    \Later{\kappa}{\phi}
    \Rightarrow
    \Later{\kappa}{\psi}
    \tag{${\land}\dashv{\Rightarrow}$}
    \\
    &\equiv
    \Later{\kappa}{
      \Parens*{
        \Parens*{\phi\Rightarrow\psi}
        \land
        \phi
      }
    }
    \Rightarrow
    \Later{\kappa}{\psi}
    \tag{Theorem~\ref{thm:later-cartesian}}
  \end{align*}

  Now, assuming
  $\Later{\kappa}{\Parens{\Parens{\phi\Rightarrow\psi}\land\phi}}$, we
  have to show $\Later{\kappa}{\psi}$. Observe that
  $\Parens{\Parens{\phi\Rightarrow\psi}\land\phi}\Rightarrow\psi$;
  therefore, by monotonicity (Corollary~\ref{cor:later-monotone}) we
  have $\Later{\kappa}{\psi}$, which was our goal.

  ($\Leftarrow$) We will reason externally through
  Lemma~\ref{lem:kj-forall-implies}; fixing a world $\mathbf{U}$ and
  elements $\Member{\kappa}{\ClkObj\Parens{\mathbf{U}}}$,
  $\Member{\phi,\psi}{\Omega\Parens{\mathbf{U}}}$ such that
  $\Forces{\mathbf{U}}{\Later{\kappa}{\phi}\Rightarrow\Later{\kappa}{\psi}}$, we need
  to show that
  $\Forces{\mathbf{U}}{\Later{\kappa}{\Parens{\phi\Rightarrow\psi}}}$. Proceed by case
  on $\partial_U(\kappa)$:
  \begin{proofcases}
  \item[$\IsEq{\partial_U(\kappa)}{0}$] Immediate.
  \item[$\IsEq{\partial_U(\kappa)}{n+1}$] Now we need to show:
    \[
      \Forces{
        \mathbf{U}[\kappa\mapsto n]
      }{
        \Squares{\kappa\pluseq1}^*\phi\Rightarrow\Squares{\kappa\pluseq1}^*\psi
      }
  \]

  Fix $\Of{\rho}{\mathbf{V}\to\mathbf{U}[\kappa\mapsto n]}$ such
  $\Forces{\mathbf{V}}{\rho^*\Squares{\kappa\pluseq1}^*\phi}$ to show that
  $\Forces{\mathbf{V}}{\rho^*\Squares{\kappa\pluseq1}^*\psi}$.
  Writing $\mathbf{V}'$ for
  $\mathbf{V}[\rho^*\kappa\mapsto \partial_V(\rho^*\kappa)+1]$,
  observe that we can form a map
  $\Of{\sigma}{\mathbf{V'}\to\mathbf{U}}$ such that the following
  diagram commutes:
  \[
    \begin{tikzcd}[sep=huge]
      \mathbf{V}
      \arrow[r,"\rho"]
      \arrow[d,swap,"\Squares{\rho^*\kappa\pluseq1}"]
      &
      \mathbf{U}[\kappa\mapsto n]
      \arrow[d,"\Squares{\kappa\pluseq1}"]
      \\
      \mathbf{V'}
      \arrow[r,swap,dashed,"\sigma"]
      &
      \mathbf{U}
    \end{tikzcd}
  \]
  As a map in $\FINPlus$, $\sigma$ is the same as $\rho$; to see that
  it is a map in $\CLK$, observe that $m_1 + 1\leq m_2 + 1$ iff
  $m_1\leq m_2$.
  Now, we have assumed
  $\Forces{\mathbf{U}}{\Later{\kappa}{\phi}\Rightarrow\Later{\kappa}{\psi}}$;
  instantiating this assumption at $\sigma$, we have the following
  external implication:
  \[
    \Forces{\mathbf{V}'}{\Later{\sigma^*\kappa}{\sigma^*\phi}}
    \Rightarrow
    \Forces{\mathbf{V}'}{\Later{\sigma^*\kappa}{\sigma^*\phi}}
  \]

  Observing that the action of $\sigma$ on $\kappa$ is the same as the
  action of $\rho$ on $\kappa$ (since $\ClkObj$ is oblivious to time
  assignments), we can unfold our implication further:
  \[
    \Forces{\mathbf{V}}{\Squares{\rho^*\kappa\pluseq1}^*\sigma^*\phi}
    \Rightarrow
    \Forces{\mathbf{V}}{\Squares{\rho^*\kappa\pluseq1}^*\sigma^*\psi}
  \]

  By the diagram above, we calculate the composition of reindexings:
  \[
    \Forces{\mathbf{V}}{\rho^*\Squares{\kappa\pluseq1}^*\phi}
    \Rightarrow
    \Forces{\mathbf{V}}{\rho^*\Squares{\kappa\pluseq1}^*\psi}
  \]

  But we have already assumed
  $\Forces{\mathbf{V}}{\rho^*\Squares{\kappa\pluseq1}^*\phi}$, and
  $\Forces{\mathbf{V}}{\rho^*\Squares{\kappa\pluseq1}^*\psi}$ is what we
  were trying to prove.

  \end{proofcases}
\end{proof}

\begin{theorem}[L\"ob induction]\label{thm:loeb-induction}
  We have the following \emph{L\"ob induction} principle for the later
  modality:
  \[
    \ClkForall{\kappa}{
      \forall\Of{\phi}{\Omega}.\
      \Parens*{
        \Later{\kappa}{\phi}\Rightarrow\phi
      }
      \Rightarrow \phi
    }
  \]
\end{theorem}
\begin{proof}
  By Lemma~\ref{lem:kj-forall-implies}, it suffices to show that for
  all $\mathbf{U}$ and
  $\Member{\kappa}{\ClkObj\Parens{\mathbf{U},\kappa}}$, we have the
  external proposition $P\Parens{\mathbf{U},\kappa}$, defined as follows:
  \[
    P\Parens{\mathbf{U,\kappa}} \triangleq
    \forall\Member{\phi}{\Omega\Parens{\mathbf{U}}}.\
    \Parens{
      \Forces{\mathbf{U}}{
        \Later{\kappa}{\phi}
        \Rightarrow
        \phi
      }
    }
    \Rightarrow
    \Forces{\mathbf{U}}{\phi}
  \]

  We proceed by induction on $\partial_U(\kappa)$; in what follows, we
  will write $\mathbf{U}_n$ for $\mathbf{U}[\kappa\mapsto n]$.
  \begin{proofcases}
  \item[$\IsEq{\partial_U(\kappa)}{0}$] We need to establish
    $P(\mathbf{U}_0,\kappa)$. Fix
    $\Member{\phi}{\Omega\Parens{\mathbf{U}_0}}$ such that
    $\Forces{\mathbf{U}_0}{\Later{\kappa}{\phi}\Rightarrow\phi}$, to
    show $\Forces{\mathbf{U}_0}{\phi}$. Instantiating our assumption
    with the identity morphism, it suffices to show that
    $\Forces{\mathbf{U}_0}{\Later{\kappa}{\phi}}$; but this is
    trivial, since the value of $\kappa$ is $0$.
  \item[$\IsEq{\partial_U(\kappa)}{n+1}$] Our induction hypothesis is
    $P(\mathbf{U}_n,\kappa)$, and we need to show
    $P(\mathbf{U}_{n+1},\kappa)$. Fix
    $\Member{\phi}{\Omega\Parens{\mathbf{U}_{n+1}}}$ such that
    $\Forces{\mathbf{U}_{n+1}}{\Later{\kappa}{\phi}\Rightarrow\phi}$,
    to show $\Forces{\mathbf{U}_{n+1}}{\phi}$. Instantiating this
    assumption with the identity morphism, it suffices to show
    $\Forces{\mathbf{U}_{n+1}}{\Later{\kappa}{\phi}}$, which is the
    same as
    $\Forces{\mathbf{U}_n}{\Squares{\kappa\pluseq1}^*\phi}$. To
    establish this, we instantiate our induction hypothesis with
    $\Squares{\kappa\pluseq1}^*\phi$, and it remains to show
    $\Forces{\mathbf{U}_n}{\Later{\kappa}{\Squares{\kappa\pluseq1}^*\phi}\Rightarrow\Squares{\kappa\pluseq1}^*\phi}$. We
    have assumed
    $\Forces{\mathbf{U}_{n+1}}{\Later{\kappa}{\phi}\Rightarrow\phi}$,
    so by reindexing we have
    $\Forces{\mathbf{U}_n}{\Later{\Squares{\kappa\pluseq1}^*\kappa}{\Squares{\kappa\pluseq1}^*\phi}\Rightarrow\Squares{\kappa\pluseq1}^*\phi}$. This
    is the same as our goal, because
    $\IsEq{\Squares{\kappa\pluseq1}^*\kappa}{\kappa}$.
  \end{proofcases}
\end{proof}

\begin{definition}[Totality]\label{def:totality}
  An object $\Of{X}{\ClkTopos}$ is called \emph{total} if its action
  on all restriction maps $[\kappa\pluseq{n}]$ is a surjection.\footnote{This
  is the analogous condition to the one described
  in~\citet{birkedal-mogelberg-schwinghammer-stovring:2011},
  generalized to the case of multiple clocks.}
\end{definition}

\begin{definition}[Inhabitedness]\label{def:inhabitedness}
  An object $\Of{X}{\ClkTopos}$ is called \emph{inhabited} when the
  formula $\exists x:X.\ \top$ is valid in the internal logic of
  $\ClkTopos$.
\end{definition}

The constant objects (such as $\Nat$) are all total; but note that an
object may be \emph{total} without being \emph{constant}: for
instance, the subobject classifier is total. In our development, we
have only needed the fact that $\Nat$ is total.

\begin{theorem}\label{thm:total-yank-existential}
  Suppose that an object $\Of{Y}{\ClkTopos}$ is total and inhabited
  (Definitions~\ref{def:totality},\ref{def:inhabitedness}). Then, if we
  \emph{later} have an element of $Y$ that satisfies $\phi$, we can
  also \emph{now} exhibit an element of $Y$ that \emph{later}
  satisfies $\phi$.
  \[
    \forall\Of{\kappa}{\ClkObj}.\
    \forall\Of{\phi}{\Omega^Y}.\
    \Later{\kappa}{
      \Parens*{
        \exists\Of{y}{Y}.\
        \phi(y)
      }
    }
    \Rightarrow
    \exists\Of{y}{Y}.\
    \Later{\kappa}{
      \phi(y)
    }
  \]
\end{theorem}
\begin{proof}
  Using Lemma~\ref{lem:kj-forall-implies}, fix a world $\mathbf{U}$ and a
  predicate $\Member{\phi}{\Omega^Y\Parens{\mathbf{U}}}$ such that
  $\Forces{\mathbf{U}}{\Later{\kappa}{\Parens{\exists\Of{y}{Y}.\
        \phi(y)}}}$; we need to show
  $\Forces{\mathbf{U}}{\exists\Of{y}{Y}.\ \Later{\kappa}{\phi(y)}}$.
  Proceed by case on $\partial_U(\kappa)$:
  \begin{proofcases}
  \item[$\IsEq{\partial_U(\kappa)}{0}$] Then it suffices to exhibit an
    arbitrary element of $Y$ at $\mathbf{U}$, since the predicate is
    trivial at this world. But we have already assumed $Y$ to be
    inhabited, so we are done.

  \item[$\IsEq{\partial_U(\kappa)}{n+1}$] In this case, our assumption
    amounts to the following external existential:
    \[
      \Forces{\mathbf{U}[\kappa\mapsto n]}{
        \exists\Of{y}{Y}.\
        \Squares{\kappa\pluseq1}^*\phi(y)
      }
    \]

    Unfolding the forcing clause for existential quantification, this
    means that we have an element
    $\Member{\alpha}{Y\Parens{\mathbf{U}[\kappa\mapsto n]}}$ such that
    the following holds:
    \[
      \Forces{\mathbf{U}[\kappa\mapsto n]}{
        \Squares{\kappa\pluseq1}^*\phi(\alpha)
      }
      \tag{H}
    \]

    Our goal was to show that
    $\Forces{\mathbf{U}}{\exists\Of{y}{Y}.\ \Later{\kappa}{\phi(y)}}$;
    because $Y$ is total, from $\alpha$ we can get an element
    $\Member{\beta}{Y\Parens{\mathbf{U}}}$ such that
    $\IsEq{\alpha}{\Squares{\kappa\pluseq1}^*\beta}$.

    Now it remains only to show that
    $\Forces{\mathbf{U}}{\Later{\kappa}{\phi(\beta)}}$; at this world, this
    is the same as to say that
    $\Forces{\mathbf{U}[\kappa\mapsto
      n]}{\Later{\kappa}{\Squares{\kappa\pluseq1}^*\phi\Parens{\Squares{\kappa\pluseq1}^*\beta}}}$. Because
    $\IsEq{\alpha}{\Squares{\kappa\pluseq1}^*\beta}$, this is the same
    as (H).
  \end{proofcases}
\end{proof}

\ifreport%
\else%
\section{Construction of \ClockCTT{}}\label{sec:appendix:attic}
In this appendix, we present the full versions of the definitions
which were truncated in the main body of the paper.

In Figures~\ref{fig:full-programs} and~\ref{fig:full-opsem}, we show
the unabridged inductive definition of programs $\ITm{n}$ and their
operational semantics; in Figure~\ref{fig:full-term-elab}, we give the
full definition of the program elaboration function $\Sem{-}$;
finally, in Figure~\ref{fig:full-monotone-operator} we give the full
definition of the monotone operator $\TSFun{\sigma}$ on candidate type
systems.

\begin{figure*}
  \begin{minipage}[c]{1.0\linewidth}
    \input{figures/programs}
  \end{minipage}
  \caption{The full inductive definition of the programs with $n$
    free variables
    $\Of{\ITm{n}}{\ClkTopos}$.}\label{fig:full-programs}
\end{figure*}

\begin{figure*}
  \begin{minipage}[c]{1.0\linewidth}
    \input{figures/opsem}
  \end{minipage}
  \caption{The full small-step operational semantics for closed
    \ClockCTT{} programs.}\label{fig:full-opsem}
\end{figure*}

\begin{figure*}
  \begin{minipage}[c]{1.0\linewidth}
    \input{figures/term-elab}
  \end{minipage}
  \caption{The full definition of the program elaboration function $\Sem{-}$.}\label{fig:full-term-elab}
\end{figure*}

\begin{figure*}
  \begin{minipage}[c]{1.0\linewidth}
    \input{figures/monotone-operator}
  \end{minipage}
  \caption{The full definition of the monotone operator on candidate
    type systems, $\TSFun{\sigma}$.}\label{fig:full-monotone-operator}
\end{figure*}

\fi%

\ifreport%
\else%
\section{Validated Rules}
Finally, we present the full set of inference rules for \ClockCTT{}
which we have verified in our Coq formalization.

\fi%



\nocite{atkey-mcbride:2013}
\bibliographystyle{plainnat}
\bibliography{references/refs}

\end{document}